\documentclass[journal,draftclsnofoot,onecolumn,12pt]{IEEEtran}

\usepackage{color}
\usepackage{graphicx}
\usepackage{verbatim}
\usepackage{amsmath}
\usepackage{amssymb}
\usepackage{amsthm}
\usepackage{algorithm}
\usepackage{algpseudocode}
\usepackage{hyperref}
\usepackage{cite}
\usepackage[colorlinks]{}
\ifCLASSINFOpdf
\else
\fi
\newtheorem{proposition}{Proposition}
\newtheorem{theorem}{Theorem}
\newtheorem{corollary}{Corollary}
\newtheorem{lemma}{Lemma}

\newtheorem{definition}{Definition}
\newtheorem{remark}{Remark}
\newtheorem{example}{Example}

\makeatletter
\def\BState{\State\hskip-\ALG@thistlm}
\makeatother

\begin{document}

\def\N{\mathcal{N}}
\def\As{A_{\mathbf{s}}}
\def\Bs{B_{\mathbf{s}}}
\def\Cs{C_{\mathbf{s}}}
\def\SN{S_{\mathcal{N}}}
\def\w{\mathbf{w}}
\def\z{\mathbf{z}}
\def\c{\mathbf{c}}
\def\u{\mathbf{u}}
\def\h{\mathbf{h}}
\def\g{\mathbf{g}}
\def\f{\mathbf{f}}
\def\0{\mathbf{0}}
\def\s{\mathbf{s}}
\def\Aw{A_{\mathbf{w}}}
\def\Bw{B_{\mathbf{w}}}
\def\Cw{C_{\mathbf{w}}}
\def\M{\mathcal{M}}
\def\C{\mathcal{C}}
\def\I{\mathcal{I}}
\def\Cd{\mathcal{C}^\perp}
\def\FF{\mathbb{F}}
\def\FFq{\mathbb{F}_q}
\def\FFqn{\mathbb{F}_{q}^{N}}
\def\FFqk{\mathbb{F}_{q}^{k}}
\def\lambdaC{\Lambda_{\mathcal{C}}}
\def\lambdaCd{\Lambda_{\mathcal{C}^{\perp}}}

\def\support{{\Lambda}}
\def\br{{\bf r}}
\def\bc{{\bf c}}
\def\bz{{\bf z}}
\def\bs{{\bf s}}
\def\bu{{\bf u}}
\def\bv{{\bf v}}
\def\bg{{\bf g}}
\def\bff{{\bf f}}
\def\ground{{\Omega}}
\def\bc{{\bf c}}
\def\solA{{\bf A}}
\def\solB{{\bf B}}
\def\solC{{\bf C}}
\def\field{{\mathbb F}_{q}}
\def\Kp{{\mathcal K}}

\def\N{{\cal N}}
\def\bz{{z_{\scriptscriptstyle \N}}}
\def\by{{y_{\scriptscriptstyle \N}}}
\def\l{\left}
\def\r{\right}
\def\gf{ {\mathbb F}}
\def\rankfn{\rho}
\def\A{\mathcal A}
\def\M{\mathcal M}
\def\X{\mathcal X}
\def\B{\mathcal B}
\def\reals{\mathbb R}
\def\Z{\mathcal Z}
\def\K{\mathcal K}
\def\C{\mathcal C}
\def\p{\prime}
\def\real{{\mathbb R}}
\def\dist{{W}}

\def\x{{\bf x}}
\def\y{{\bf y}}

\def\pattern{{\gamma}}
\def\erasurematrix{{\bf I}}

\def\Ha{{\bf H}}
\def\Hb{{\bf K}}

\def\X{{\bf X}}
\def\H{{\bf H}}
\def\solG{{\bf G}}
\def\zero{{\bf 0}}
\def\r{{\bf r}}

\def\R{\mathcal R}
\def\P{\mathcal P}
\def\cbP{\mathbf P}
\def\FN{\mathbf {FailedNodes}}
\def\FP{\mathbf {FailurePatterns}}
\def\A{\mathcal{A}} 
\def\F{\mathcal{F}} 
\def\G{\mathcal{G}} 
\def\S{\mathcal{S}} 
\def\J{\mathcal{J}} 
\def\E{\mathcal{E}} 
\def\HH{\mathcal{H}}
\def\j{\mathbf{J}}
\def\L{\mathcal{L}}
\def\U{\mathcal{U}}
\def\q{\mathbf{q}}
\def\v{\mathbf{v}}
\def\D{\mathcal{D}}
\def\Q{\mathcal{Q}}
\def\d{\mathbf{d}}
\def\k{\mathbf{k}}
\def\b{\mathbf{b}}
\def\solS{\mathbf{S}}
\def\solF{\mathbf{F}}

\def\lmd{{{\textsf {lmd}}}}
\def\gmd{{{\textsf {gmd}}}}

\def\G{{G}}
\def\bx{{\bf x}}
\def\by{{\bf y}}
\def\bg{{\bf g}}
\def\zeroz{\textbf{0}}
\def\X{{X}}


\algnewcommand{\algorithmicgoto}{\textbf{go to}}%
\algnewcommand{\Goto}[1]{\algorithmicgoto~\ref{#1}}%

\title{Multi-Rack Distributed Data Storage Networks}

\author{\IEEEauthorblockN{Ali Tebbi\IEEEauthorrefmark{2},
Terence H. Chan\IEEEauthorrefmark{2},
Chi Wan Sung\IEEEauthorrefmark{3}} 

\IEEEauthorblockA{\IEEEauthorrefmark{2} Institute for Telecommunications Research, University of South Australia\\
Email: \{ali.tebbi, terence.chan@unisa.edu.au\}}

\IEEEauthorblockA{\IEEEauthorrefmark{3} Department of Electronic Engineering, City University of Hong Kong\\
Email: albert.sung@cityu.edu.hk}

\thanks{An earlier conference version of the paper appeared at ITW 2014 \cite{RMSC}. In this journal version, we present more details on the repair processes of failures and also establish linear programming bounds on the code size.}
}

\maketitle

\begin{abstract}


The majority of works in distributed storage networks assume a simple network model with a
collection of identical storage nodes with the same communication cost between the nodes. In this paper,
we consider a realistic multi-rack distributed data storage network and present a code design framework
for this model. Considering the cheaper data transmission within the racks, our code construction method
is able to locally repair the nodes failure within the same rack by using only the survived nodes in
the same rack. However, in the case of severe failure patterns when the information content of the
survived nodes is not sufficient to repair the failures, other racks will participate in the repair process.
By employing the criteria of our multi-rack storage code, we establish a linear programming bound on
the size of the code in order to maximize the code rate.

\end{abstract}


%
\IEEEpeerreviewmaketitle

\section{Introduction}
\IEEEPARstart{M}{ost} of the existing distributed storage network models assume a very simple structure that the network itself is viewed as a collection of \emph{identical} storage nodes and that the transmission cost between any two nodes are identical \cite{NCDSS,DScodeRepairTransfer,ExactMDScodeIA,LocalReg,LRCDimakis}. However, this model cannot perfectly represent the real world storage networks. In reality, a typical data centre can easily house hundreds of racks each of which contains numerous storage disks \cite{DSNArchProtManag}. While all storage nodes (or disks here) in the network can communicate with each other, the transmission costs in terms of latency or overheads can differ vastly. For example, the transmission latency between storage disks in the same rack is usually much smaller, when compared with the case that both disks are not in the same rack \cite{DiskLocal}. It is reported that in practical networks the inter-rack communication cost is typically $5$ to $20$ times higher than the intra-rack transmission cost \cite{ShuffleWatcher}.

A common approach of storing data in multi-rack storage networks is storing each encoded symbol of a data block in distinct nodes located in distinct racks \cite{WinAzur,XorElephant}. Consequently, repairing any node failure requires transferring data from survived nodes across the racks. Due to the large amount of data which is required to be communicated across racks during a failure repair process, this approach could be highly costly \cite{NetCharDataCener}.

While \emph{Regenerating Codes} \cite{NCDSS,SimpleRegCode,OptimalExactCode} were introduced to reduce the repair bandwidth, \emph{Locally Repairable Codes (LRC)} \cite{LRCGopalan,LRCDimakis,LRCTopo,LRmulti,RLLC,LPLR,CoopLocal,Pyramid} were introduced in order to optimise the node failure repair by involving a small group of helper nodes. It is worth to mention that any storage code, specially locally repairable codes which are designed for generic model of distributed storage networks can also be used for the multi-rack networks. However, these coding schemes are not able to provide optimal methods to repair the failures in the network since they do not take into consideration the different transmission cost between the nodes in the multi-rack storage networks. For instance, consider the locally repairable code known as \emph{Pyramid Code} \cite{Pyramid} which is used in Windows Azur storege \cite{WinAzur}. This pyramid code, as illustrated in Figure \ref{fig:pyramid}, is a $(n=16,k=12,r=6)$ locally repairable code with two \emph{local repair groups} of size $r=6$. The local repair groups are $X_{1,1}-X_{1,6}$ with its local parity node $P_{X_1}$ and $X_{2,1}-X_{2,6}$ with its local parity node $P_{X_2}$. Parity nodes $Q_1$ and $Q_2$ are global parities. Assume each repair group (and one of the global parity nodes to keep the load balanced) is stored in a separate rack.  
\begin{figure}[t]
\centering
\includegraphics[width=.5\textwidth]{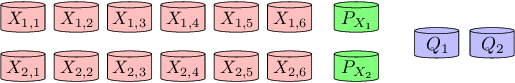}
\caption{A (16,12) pyramid code used in Windows Azur storage.}
\label{fig:pyramid}
\end{figure} 
Any node failure can be repaired inside the rack by the $r=6$ surviving nodes. However, if there exists multiple failures within a rack, the content of the helper nodes from the other rack needs to be transmitted across to the failed rack. More precisely, at least 6 symbols need to be transferred across the racks in order to repair a single failure. This will impose a high repair bandwidth to the storage network.

\begin{figure}[t]
\centering
\includegraphics[width=.4\textwidth]{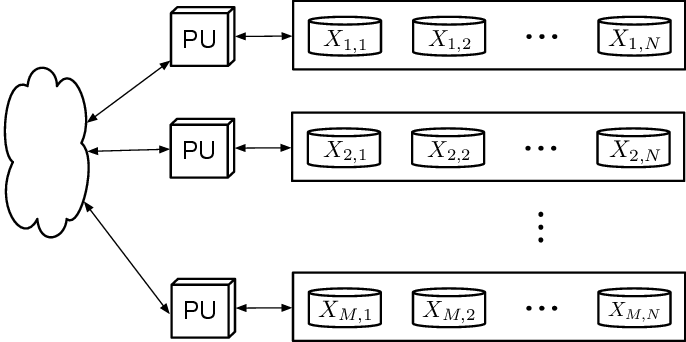}
\caption{The multi-rack distributed storage network. Each rack is equipped with a processing unit (PU) which is responsible for computations and intra-rack and inter-rack communication.}
\label{fig:rack}
\end{figure} 

In this paper we introduce a realistic multi-rack storage network which represents the real data storage networks more generally and practically. We focus on a storage code-design framework, specifically tailored for multi-rack data storage networks and their requirements. Our storage network model depicted in Figure \ref{fig:rack} consists of $M$ racks each of which contains $N$ storage nodes. We will assume that each rack has a \emph{Processing Unit} (PU) which is directly connected to all storage nodes in the same rack. 
It is worth mentioning that in practical data centres all servers in a rack are connected to an in-rack switch which is called Top-of-Rack (ToR) switch. The ToR switch is responsible for the intra-rack communication while one or more servers in the rack could be used as compute nodes \cite{CiscoDesignI}. This architecture can be viewed as the Processing Unit of the rack.
The rack processing unit is responsible for both computation on the stored data and communication between the nodes in the rack. Moreover, the processing units of racks can communicate to each other in order to transmit data from one rack to another (via \emph{aggregation switches} \cite{CiscoDesignI}). In other words, storage nodes in two different racks can only communicate via their respective processing units. It is very common in realistic systems that the communication cost between the storage nodes within a rack (via its PU) is much lower than the communication cost between two different processing units (and hence located in two different racks) \cite{DiskLocal}. Therefore, it is desirable and in fact critical that a failed node could be repaired by only the survived nodes within the same rack in order to keep the repair cost low. We further assume that the system bottleneck is at the PU of each rack. Therefore, in our code-design framework, the focus is to design distributed storage codes that minimise the communication costs between nodes and the processing unit of the rack. It is only for some occasional severe failure patterns that it will require nodes from other racks to assist in the repair process. 

Our multi-rack storage code is defined with three parity check matrices $\Ha$, $\Hb$, and $\solG$. Matrix $\Ha$ determines the intra-rack repair groups such that any failed node in any of the racks can be repaired by at most $r_1$ surviving nodes in the same rack. We derive the conditions on $\H$ such that the intra-rack local repair can still be successful even in the presence of multiple failures in the rack. However, as mentioned earlier, in the case of severe failure patterns where the intra-rack repair fails, matrices $\Ha$ and $\Hb$ determine a group of helper nodes from the same rack of the failure and the other racks in order to proceed with the inter-rack repair. Moreover, in our coding scheme, parity matrix $\solG$ determines the group of helper racks during an inter-rack repair. We show that $\solG$ can be designed separately as the parity check matrix of a locally repairable code such that only a small group of racks participate in the inter-rack repair process.

In general, existing locally repairable coding schemes are not suitable candidates for practical multi-rack distributed storage networks. Assume that a repair group of an LRC is considered to be a rack. In the case that a local repair fails, the network will need other repair groups (racks) to help for the failure repair. This assumption is costly due to the geographically distributed nature of the storage networks. However, in our coding scheme, it is only in occasional sever failure patterns that the other racks are needed to help repairing the failure. Moreover, the existing LRC schemes do not take into account the different communication cost between the nodes.  A major advantage of our coding scheme compared to similar schemes, such as LRCs, is the concept of the rack processing unit. In contrast to LRCs, our coding scheme enables the helper racks to only transmit a linear combination of the helper nodes' content to the failed rack via their processing unit. This approach optimises the inter-rack repair bandwidth and significantly reduces the inter-rack communication cost.

\subsection{Related Work}

A regenerating code \cite{NCDSS} is proposed in \cite{HierarchicalDataCenter} in order to minimise the across rack repair bandwidth. In this coding scheme, each rack stores multiple encoded symbols (rather than one symbol \cite{XorElephant}) of a data block in distinct storage nodes. To repair a failed node, first a regeneration will be occurred within each rack (i.e., one of nodes in each rack collects the encoded data from all nodes in the rack and re-encodes it) and then the regenerated data from each rack will be transferred to the failed rack to regenerate the content of the failed node. 
Recently, heterogeneous distributed storage networks (including the multi-rack models) has received a fair amount of attention \cite{RackModelDSN,Gaston2013A-Realistic,NonHomoRack,NonHomogenDSN,CostBandwith} due to the heterogeneous nature of the practical storage networks and their various applications such as hybrid storage systems \cite{HybridOceanstore}, video-on-demand systems \cite{p2pvod}, and heterogeneous wireless networks \cite{WirelessD2D}. A  heterogeneous model for distributed storage networks is introduced in \cite{CostBandwith} where a static classification of the storage nodes is proposed. In this model the storage nodes are partitioned into two groups with "cheap" and ''expensive'' bandwidth. In other word, the data download cost, to repair a failure, from the nodes in the "cheap bandwidth" group will be lower than the data download cost from the "expensive bandwidth" group. The model in \cite{CostBandwith} partially addressed the issues that the communication costs among nodes are not all equal where the download cost from one of the groups is always cheaper than the other group. However, this model does not fit well into a multi-rack model, where the transmission cost should depend on both where the transmitting and the receiving (or the failed) storage nodes are located.

A more realistic rack model of a distributed storage network has been investigated in \cite{Gaston2013A-Realistic}, in which the authors considered a two-rack model. In their model, the communication cost between the nodes in the same rack is smaller than between two different racks. Therefore, the main difference of this model compared to the one in \cite{CostBandwith} is that the classification of the storage nodes depends on the location of the failed node. More precisely, the data download cost from the nodes in the same rack (i.e., group) of the failure is lower (i.e., cheap bandwidth) than the download cost from the other group (i.e., expensive bandwidth). As such, it is desirable that more data should be transmitted by nodes in the same rack of the failed node during the repair process. Using an information flow graph, \cite{Gaston2013A-Realistic} derives the trade-off between the storage cost and repair cost by identifying that if certain choice of parameters are achievable or not. The trade-off in \cite{CostBandwith} and \cite{Gaston2013A-Realistic} is asymptotic (without restriction on the alphabet size) and functional repair is always assumed (i.e., the failed node is not required to recover exactly what it was previously storing, as long as the whole storage system is still robust after repair).

A non-homogenous storage system is considered in \cite{NonHomogenDSN} where there exists a super node in the network with higher storage capacity, reliability and availability probability than the other nodes. It has been shown that this model can achieve the optimal bandwidth-storage trade-off bound in \cite{NCDSS} with a smaller file and alphabet size than the traditional homogeneous storage network in \cite{RepairHadamard}.

The Data retrieval problem in heterogeneous storage systems is studied in \cite{DSAllocate}. In this model it is assumed that each node has a different storage size where any amount of encoded data can be stored in each node such that the total allocated storage remains less than a threshold. The optimal allocation to retrieve the original data is studied such that the data collector can access to only a random group of nodes. A combination of the repair problem with data allocation is investigated in \cite{RepairAlloc1} and \cite{RepairAlloc}. In these works a general model of a heterogeneous storage network is considered where each node has a different storage and download cost. The amount of data allocated to each storage node and the amount of data to be downloaded from each survived node to repair a failure has been investigated using the information flow graph to minimise the storage and repair cost and establishing a storage-repair trade-off.

The capacity of heterogenous storage networks is studied in \cite{CapacityHeterogeneous}. The proposed network in this work consists of storage nodes with different storage capacities and repair bandwidths. It is assumed that the repair bandwidth of each node depends on the repair group that the helper nodes belong to. The functional repair of node failures is assumed and the capacity of this network as the maximum amount of stored information in order to reach a level of reliability is studied. 

Block Failure Resilient (BFR) codes are studied in \cite{RBFR}. The authors consider a distributed storage network with a single failure domain \cite{AvailGlob} where the storage nodes are divided in blocks (e.g. racks). The failure of a block will result in unavailability of the nodes in that block. Consider a storage network with $n$ nodes and $b$ blocks where each block contains $\frac{n}{b}$ nodes. BFR codes relax the node-repairability and data-reconstruction constraints of the regenerating codes such that any failed node within a block can be repaired by contacting any $d_r$ nodes of any $b_r=b-\sigma$ available blocks (i.e., $d=d_r b_r$ nodes in total). Moreover, the original data can be retrieved by contacting any $k_c$ nodes of any  $b_c=b-\rho$ available blocks (i.e., $k=k_c b_c$ nodes in total) where $\rho$ is the resilience parameter. For such a relaxation, similar to the regenerating codes, the storage per node and repair bandwidth trade-off is derived. Locally repairable BFR codes are also introduced in \cite{RBFR} such that a failed node can be repaired by contacting the nodes of a local group of blocks (e.g., cluster). One of the main differences of the network model in \cite{RBFR} with our model is that it is assumed that always during the repair process of a failed node, the other nodes of the same rack are also unavailable (i.e., single failure domain) and they are not able to contribute in the repair process. In our model, we assume that in non-severe failure patterns, the available nodes of the rack can locally repair the failed node.

A similar network model to our work is considered in \cite{ClusterSS_trade-off} where the network consists of $n$ clusters (e.g., racks) each of them stores $m$ nodes. The network is fully connected such that the nodes within a cluster are connected via an intra-cluster link and the clusters are connected via an inter-cluster link. The proposed coding scheme is a generalisation of the regenerating codes \cite{NCDSS} where a file of size $B$ symbols is encoded into $nm\alpha$ symbols and stored across $nm$ nodes in the network such that each node stores $\alpha$ symbols. In order to repair a failed node, $\beta$ symbols will be downloaded each from any subset of $d$ clusters. These $\beta$ symbols are a function of the content (at most $\gamma^\prime$, $\gamma^\prime \leq \alpha$ symbols) of at most $\ell^\prime$ nodes in each helper cluster. Moreover, the content (at most $\gamma$, $\gamma \leq \alpha$ symbols) of $\ell$ local helper nodes will be downloaded to contribute in the repair process. Utilising the information flow graph under the functional repair settings, an upper bound on the file size $B$ is derived. For fixed values of $B$, the bound gives the trade-off between storage and inter-cluster bandwidth. A lower bound on intra-cluster bandwidth $\gamma$ is also obtained. Unlike our coding scheme where the inter-rack repair happens only in the case of severe failure patterns, a failed node in \cite{ClusterSS_trade-off} always is repaired by the help of a group of $d$ clusters (racks). Note that, all the bounds obtained in the aforementioned papers \cite{RBFR,ClusterSS_trade-off} are based on the information flow graph under functional repair settings.

The capacity of clustered storage systems is investigated in \cite{CapacityCluster}. The proposed network model consists of $n$ storage node distributed over $L$ clusters each of which contains $n_I=\frac{n}{L}$ nodes each with storage size $\alpha$. A failed node is regenerated by downloading $\beta_I$ symbols each from $d_I$ nodes within the same rack and $\beta_c$ symbols each from $d_c$ nodes from each cluster. It is assumed that during the repair process, all other nodes are available and will be contacted (i.e., $d_I=n_I-1$ and $d_c=n-n_I$). Also $\beta_I \geq \beta_c$, due to the lower inter-cluster communication bandwidth compared to the intra-cluster.  Employing the information flow graph, the storage capacity of this network is obtained in terms of the node storage size $\alpha$, intra-cluster repair bandwidth $\gamma_I=d_I \beta_I$, and inter-cluster repair bandwidth $\gamma_c=d_c \beta_c$. Note that since the coding scheme is based on regenerating codes, in order to minimise the repair bandwidth all $n-1$ nodes need to help to repair the failure.

The availability of clustered storage networks is studied in \cite{AvailablityCluster}. The aim in this work is to partition $n$ storage nodes of the network into $s$ clusters of size $d$ (there could be an extra cluster of size $< d$) such that any failed node in a cluster can be repaired by any of the remaining clusters (except the last cluster with less storage nodes) as its repair group. Then, the network is said to have availability $s-1$. The objective is achieving high availability and low repair bandwidth. The storage per node vs repair bandwidth trade-off is characterised following the network information flow graph under the functional and exact repair settings. Some class of codes are also proposed to minimise the exact repair bandwidth.

The notion of codes with \emph{hierarchical locality} has been studied in \cite{HierarchicalLocality}. Codes with hierarchical locality are an extension on the codes with $(r,\delta)$-locality which are introduced in \cite{LocalReg} such that any code symbol can be recovered locally by at most $r$ other symbols even in the presence of an additional $(\delta-2)$ erasures. A $h$-level hierarchical code is an $[n,k,d]$ linear code $\C$ with locality parameters $[(r_1,\delta_1),(r_2,\delta_2),\ldots,(r_h,\delta_h)]$ where depending on the number of the failures (i.e., $\delta_i-1$), there exists a punctured code $C_i$ with locality parameters $(r_i,\delta_i)$ that can repair the failures.

The fact that coding at large lengths allows better error-tolerance for a given overhead, motivated the work in \cite{GridTopology}. One of the main challenges in the storage networks with large length codes is correlated failures which could happen due to e.g. a rack failure, a data centre failure, or failure of a power source shared by a group of servers (i.e., single failure domain \cite{AvailGlob}). This work views the code design for a distributed storage network as a two step process of 1)picking a topology and 2) optimising encoding/decoding efficiency and maximising reliability. The authors consider a simple grid-like topology (which is also extendable to the multi-rack storage networks) where each row and column of coded symbols has a bunch of parity equations and there are some global parity equations that depend on all symbols (i.e., tensor products of row and column codes, augmented with global parity equations). A lower bound on the field size of the Maximally Recoverable codes is obtained and the correctable erasure patterns by these codes are characterised. An asymptotically optimal family of Maximally Recoverable codes for one basic topology is also proposed.  

Despite of the applications of the works in \cite{HierarchicalLocality} and \cite{GridTopology} in multi-rack storage networks, neither of them propose a general code design framework which is specifically tailored for multi-rack distributed storage network considering various network parameters such as different intra-rack and inter-rack communication cost. For example, assume that a storage code with hierarchical locality is employed in a multi-rack storage network. Depending on the failure pattern, it would need all the racks to be available to repair a failure which is not a practical assumption due to geographically distributed nature of the network. Moreover, in \cite{HierarchicalLocality} and generally other LRC schemes in the literature, it is assumed that during a repair process, the content of each helper node will be transmitted separately to the failed node's replacement (newcomer) in order to recover the lost data. In a multi-rack storage network, this will impose a high inter-rack repair bandwidth to the network due to the high inter-rack communication cost.

\subsection{Contributions and Organisation}
The main contributions of this paper are:
\begin{itemize}
\item
A code-design framework for multi-rack storage networks: 
we propose a general coding scheme for multi-rack distributed storage networks. Our proposed scheme is defined by three parity check matrices $\Ha$, $\Hb$, and $\solG$.
This coding scheme is able to locally repair any node failure within the rack by using matrix $\Ha$ in order to minimise the repair cost. Moreover, in the case of severe failure patterns that the failures cannot be repaired only by the survived nodes inside the rack, by using matrix $\Hb$, our scheme is able to engage some of the nodes in other racks in the repair process. The helper racks will be determined by matrix $\solG$.
\item
Establishing linear programming bounds on the code size: we show that maximising the rate of the multi-rack storage code is equivalent to maximising the code size. We establish a linear programming problem on the code size based on the definition and criteria of our multi-rack storage code. The maximum size of the code in turn will determine the optimal size of the parity check matrices $\Ha$ and $\Hb$.
\end{itemize}

This paper is extended from our earlier work on multi-rack distributed storage codes \cite{RMSC} which is presented in IEEE Information Theory workshop (ITW 2014). The rest of this paper is organised as follows. In Section \ref{Sec:IV} we present the code-design framework for multi-rack storage networks and give a detailed description of its criteria and the failures repair processes. We also derive the code rate in this section. Then, In Section \ref{Sec:V}, we establish a linear programming problem to upper bound the code size. Moreover, in this section, we exploit symmetry in our code in order to reduce the complexity of the problem. The paper is concluded in Section \ref{Sec:Conc}.


\section{Multi-rack Storage Code --  Design Framework} \label{Sec:IV}
In this section, we first introduce our system model and multi-rack storage code which is defined by three parity check matrices $\Ha$, $\Hb$, and $\solG$. We then describe the intra-rack repair process and show how the failures can be repaired only by the surviving nodes inside the rack using the parity check matrix $\Ha$. The inter-rack repair process will be described afterwards where we show how a failure can be repaired by the surviving nodes inside the rack and the nodes in helper racks when the intra-rack repair fails. Finally, we present the rate of the multi-rack storage code which will be used in the next section to establish an upper bound on the code size.

Consider the rack model storage network depicted in Figure \ref{fig:rack}. This multi-rack data storage network consists of $M$ racks each of which contains $N$ storage nodes (or storage disks). We will represent each node as
\[
\left(X_{m,n}, ~ \forall m \in \M \text{ and } \forall n \in \N \right)
\]
where $\M=\{1, \ldots, M\}$, $\N=\{1, \ldots, N\}$, and $X_{m,n}$ is referred to the $n$th node in the $m$th rack. Abusing notations, $X_{m,n}$ will also be referred to the content stored at that particular storage node. We define 
\[
X_{m,*} \triangleq [X_{m,1} , \ldots, X_{m,N}],
\]
whose entries are from $\FFq$. Particularly, $X_{m,*}$ is the vector of encoded data stored in the rack $m$. 
Collecting all the stored contents from each rack, we have 
\begin{align} \label{eq:codewordmatrix}
\X \triangleq
\left[
\begin{array}{ccc}
X_{1,1} & \cdots & X_{1,N} \\
\vdots & \ddots & \vdots \\
X_{M,1} & \cdots & X_{M,N} 
\end{array}
\right].
\end{align}

In this paper, we assume that each rack has a processing unit, which is responsible for all computations required in nodes repair. In other words, contents stored in a failed node will be regenerated in the processing unit, before sending all the regenerated content to the failed node (or its replica).
 
\begin{definition}[Multi-rack storage codes] \label{def:multi-rack parities}  
A {\bf{\emph{multi-rack storage code}}} is defined by  three parity check matrices $(\Ha ,  \Hb,  \solG)$ over $\FFq$ of respectively sizes $S_{1} \times N$, $S_{2} \times N$ and $L \times M$. The three matrices induce a storage code such that  $\X$ must satisfy the following parity-check equations
\begin{align} 
\Ha \X^{\top}  & = \zero \label{eq:IRP}\\ 
\Hb \X^{\top}  \solG^{\top} & = \zero.\label{eq:ARP}
\end{align}
We will call $\Ha, \Hb$ respectively the intra-rack and inter-rack parity matrices. The matrix $\solG$ will be called helper-rack parity check matrix. 
\end{definition}
 
Later in Section \ref{Sec:V} we will show that maximising the code rate is equivalent to maximising the size of the code which in turn can determine the optimal value of $S_1$ and $S_2$. Moreover, we will show that the value of $L$ is only dependent on the network topology and can be chosen separately from $S_1$ and $S_2$.   

In multi-rack storage code, it is expected that most of the node failures should be recovered and repaired locally within their own racks. However, in the special case where local repair is not possible, redundancies added among rack will be used in the recovery. As there is a much lower probability that nodes in a rack cannot be recovered locally within the rack, this paper focuses on the special case where only one rack has node failures (or that all failed nodes in other racks can be completely repaired locally). 

We now consider the first case where failures in a rack can be repaired by using only nodes within the rack.

\subsection{{Intra-Rack Repair}}

In this subsection, we will describe how to repair nodes locally within  a rack. Assume without loss of generality that rack 1 fails (i.e., a group of nodes fails inside the rack). Let $\pattern$ be the index set for the nodes in rack 1 that fail. In other words, the values of $\{X_{1, n} , n\in\pattern\}$ (i.e., the node content) are unknown to the processing unit in rack 1. 
Let 
\[
{\x} =
\left[
\begin{array}{c}
x_{1} \\ \vdots \\ x_{N}
\end{array}
\right]
\]
where 
$$
x_{n} = 
\begin{cases}
X_{1,n} & \text{ if } n \not \in \pattern \\
0 & \text{ otherwise.}
\end{cases}
$$
In other words, ${\x}$ is obtained from $X_{1,*}^{\top}$ by replacing   $X_{1, n} $  with $0$ for  all $n\in\pattern $. 

Define $\erasurematrix^{\beta}_{N}$ as an $N\times N$  diagonal matrix such that its $(n,n)^{th}$ entry is 1 if $n \in\beta$ and is 0 otherwise. For simplicity, we will drop  the subscript $N$ if it is understood from the context. Let $\bar\pattern$ be the complement set of $\pattern$. Therefore, $\bar\pattern$ will be the set of survived nodes in the rack $1$. Consequently,
$
\erasurematrix^{\bar\pattern} X_{1,*}^{\top} = {\x}
$. 
Recall that 
$
\Ha X_{1,*}^{\top}   = \zero. 
$
Therefore, rack 1 can repair ALL its failed nodes by the local rack survived nodes if and only if the following system of linear equations
\begin{align} \label{eq:21}
\begin{cases}
\erasurematrix^{\bar\pattern} X_{1,*}^{\top}   = {\x}  \\
\Ha X_{1,*}^{\top}    = \zero 
\end{cases}
\end{align}
has a unique solution. 
For notation simplicity, we will use $\langle \erasurematrix^\beta,\Ha \rangle$ to denote  the vector space spanned by rows of $\erasurematrix^\beta$ and $\Ha$. The set of linear equations in \eqref{eq:21} has a unique solution if and only if 
\begin{align}\label{eq:22}
\dim \langle \erasurematrix^{\bar\pattern}  , \Ha \rangle = N.
\end{align}

\def\dist{{\text{Dist}}}
Let $\pattern_{o}$ be the smallest set such that 
$ 
\dim \langle \erasurematrix^{\bar{\pattern_{o}}}  , \Ha \rangle < N
$.  We will denote its size $|\pattern_{o}|$ as $ \dist(\Ha) $. 
By definition, if $|\pattern| < \dist(\Ha) $, then it is sufficient to use intra-rack repair 
to repair all failed nodes. 

\begin{remark} 
It is well known that $\emph{\dist}(\Ha)$ is equal to the minimum distance of a linear code defined by the parity check matrix  $\Ha$.  
\end{remark}

\begin{definition}[support] \label{def:support}
The support $\lambda(\mathbf{v})$ of a vector $\mathbf{v} = \left[v_1, v_2, \ldots, v_N\right]$ is a subset of $\{1, 2, \ldots, N\}$ such that $i \in \lambda(\bf{v})$ if and only if $v_i \neq 0$, $\forall i \in \{1, 2, \ldots, N\}$.
\end{definition}


\begin{definition} \label{def:RepairGroup}
Consider any  matrix $\Ha$ and vector $\r$ (such that both have $N$ columns).  
For any   $j = 1, \ldots, N$, let 
\[
\Omega(\H,\r,j) = \left\{\lambda(\h) \setminus j: \h \in \langle \Ha , \r\rangle \textnormal{ and } j \in \lambda(\h) \right\}.
\]
If $\r$ is the zero vector, we will simply denote $\Omega(\H,\r,j)$ by $\Omega(\H,j)$.
\end{definition}

\begin{remark}
As we shall see, $\Omega(\H,\r,j)$ plays a fundamental role in determining whether failures in a rack can be repaired or not. Specifically,    $\Omega(\H,j)$ contains all intra-rack repair groups for   $X_{1,j} $. If there exists a set (or group) $\beta \in \Omega(\H,j)$ such that all $X_{1,\ell} $ are survived for all $\ell \in \beta $, then the failed node $X_{1,j} $ can be repaired by 
using only $X_{1,\ell} $  for all  $\ell \in \beta $. The general case where $\r$ is non-zero vector will be used in the inter-rack repair and will be explained soon.
\end{remark}

\begin{example}
Suppose $\Ha$ is the intra-rack parity check matrix and is given by
\begin{equation} \label{eq:IRPM4}
\Ha=
\left[
\begin{array}{cccccccc}
1 & 1 & 1 & 0 & 1 & 0 & 0 & 0 \\ 
1 & 1 & 0 & 1 & 0 & 1 & 0 & 0 \\
0 & 1 & 1 & 1 & 0 & 0 & 1 & 0 \\
1 & 0 & 1 & 1 & 0 & 0 & 0 & 1 \\
\end{array}
\right].
\end{equation}

From Definition \ref{def:RepairGroup}, $\Omega(\H,1)$ will be given by
\begin{multline*}
\Omega(\H,1) = \Big\{\{2,3,5\},\{2,4,6\},\{3,4,8\},\{2,7,8\},\{3,6,7\},\\
\{4,5,7\},\{5,6,8\},\{2,3,4,5,6,7,8\}\Big\},
\end{multline*}
where the entries are the index set of a group of nodes in each rack. Each subset in $\Omega(\H,1)$ denotes a intra-rack repair group for repairing $X_{1,1}$ (or $X_{m,1}$ in general).
\end{example}

\begin{lemma}\label{lem:lemma1_4}
If $\beta \in \Omega(\H,\r,j)$, then there exist vectors ${\bf y} $, ${\bf y}' $  and $a\in \FFq$ such that 
\begin{align}\label{eqlemma4}
{\bf e}_{j} = {\bf y} \Ha + a \r + {\bf y}' \erasurematrix^{\beta}  
\end{align}
where
${\bf e}_{j}=[e_{j,1} , \ldots, e_{j,N}]$ is a length $N$ row vector such that 
\begin{align}
e_{j,\ell} = 
\begin{cases}
1 & \text{ if } \ell = j  \\
0 & \text{ otherwise.}
\end{cases}
\end{align}

Conversely, if there exist vectors ${\bf y} $, ${\bf y}' $  and $a\in \FFq$ such that \eqref{eqlemma4} holds, then there exists $\alpha \subseteq \beta $ such that $\alpha \in \Omega(\H,\r,j)$.
\end{lemma}
\begin{proof}
Since $\beta \in \Omega(\H,\r,j)$, then there exists $\bu= [u_1,\ldots,u_N]$ such that
1) $ \bu = {\bf y}\Ha + a\r $ for some vector ${\bf y}$ and $a \in \FFq$, 
2) $u_{j} = 1 $ and 3) $\lambda(\bu) \setminus \{j\} = \beta$.
Let ${\bf v}   =  - \bu \erasurematrix^{\beta}  $.  
Since $\lambda(\bu) \setminus \{j\} = \beta$, 
$
{\bf v}  =  - \bu +  {\bf e}_{j}
$. 
Hence, 
\begin{align*}
{\bf e}_{j} & =  {\bf v}     +   \bu  \\
& = {\bf y}\Ha + a\r + {\bf v} \\
& = {\bf y}\Ha + a\r  -  {\bf u}\erasurematrix^{\beta}.
\end{align*}
The lemma thus follows by letting  ${\bf y}' = -  {\bf u}$. The proof of the converse is straightforward and is omitted. 
\end{proof}

Based on Lemma \ref{lem:lemma1_4},  the following theorem specifies conditions for intra-rack repair.

\begin{theorem}[Intra-rack Repair]\label{thm1}
Suppose node $j$ fails in rack  $m=1$. Let 
 $\pattern_{j}$  be the index set for all failed nodes\footnote{
 $\pattern_{j}$ can be interpreted as the set of failed nodes at the moment when 
 the node $j$ is being repaired.} (hence,  $j \in \pattern_{j}$). 
If $\beta_{j} \subseteq \{1, \ldots, N\}$ satisfies the following two criteria, 
\begin{enumerate}
\item \label{item:IRRP_1}
$\beta_{j} \in \Omega(\Ha , j)$, and
\item $\beta_{j} \cap \pattern_{j} = \emptyset$,
\end{enumerate}
then there exists $c_{j,n}$ for $n\in \beta_{j}$ such that 
\begin{align}\label{eq:thm1}
X_{1,j} = \sum_{n\in\beta_{j}} c_{j,n} X_{1,n}.
\end{align}
\end{theorem}
\begin{proof}
By Lemma \ref{lem:lemma1_4} and criterion \ref{item:IRRP_1}, there exists  
${\bf y} $ and ${\bf y}' $ such that 
\begin{align} \label{eq:evector}
{\bf e}_{j} = {\bf y} \Ha + {\bf y}' \erasurematrix^{\beta_{j}}.  
\end{align}
Hence,
\begin{align}
{\bf e}_{j} X_{1,*}^{\top} & =  ({\bf y} \Ha + {\bf y}' \erasurematrix^{\beta_{j}}) X_{1,*}^{\top} \\
& =   {\bf y} \Ha  X_{1,*}^{\top}+ {\bf y}' \erasurematrix^{\beta_{j}}  X_{1,*}^{\top} \\
& =   {\bf y}' \erasurematrix^{\beta_{j}}  X_{1,*}^{\top}, 
\end{align}
where the last equality follows from \eqref{eq:IRP}.
Finally, let 
\begin{align} \label{eq:CoefVec}
[c_{j,1} , \ldots, c_{j,N}] = {\bf y}' \erasurematrix^{\beta_{j}} .
\end{align}
As the columns of $\erasurematrix^{\beta_{j}}$ indexed by $\bar\beta_{j}$ are zero,  $c_{j,n} = 0$ if $n \not\in \beta_{j}$. Therefore, we prove the theorem. 
\end{proof}

Equation \eqref{eq:thm1} essentially defines how to regenerate the content of a failed node $X_{1,j}$ from $X_{1,n}$ for $n \in \beta_{j}$ (i.e., the nodes in its repair group). In other words, node $X_{1,j}$ is a linear combination of the nodes in its repair group where the coefficients are $c_{j,n}$. In this case, $|\beta_{j}|$ symbols are transmitted to the processing unit  in rack 1, which can then  repair the failed node $X_{1,j}$ by \eqref{eq:thm1}. Clearly, the choice of $\beta_{j}$ will affect the repair cost. It is always desirable to pick $\beta_{j}$ such that its size is as small as possible.

\begin{example}
Let $\Ha$ be the intra-rack parity check matrix over $\FF_3$ such that 
\[
\Ha = 
\left[
\begin{array}{cccc}
0 & 1 & 1 & 1 \\
1 & 0 & 1 & 2
\end{array}
\right].
\]
Assume nodes $X_{1,1}$ and $X_{1,2}$ are failed. Thus, the failure pattern will be $\pattern=\{1,2\}$. Suppose we want to repair node $X_{1,1}$. The repair groups of node $j=1$ is given by
\[
\Omega(\Ha,1)=\Big\{\{3,4\},\{2,3\},\{2,4\}\Big\}.
\]
A repair group $\beta_1 \in \Omega(\Ha,1)$ is eligible for intra-rack repair process such that $\beta_1 \cap \pattern=\emptyset$. Therefore, $\beta_1=\{3,4\}$. Moreover, we choose ${\bf y}=[0~1]$. Then, $\u=\y \Ha = [1~0~1~2]$ and ${\bf y}'=-\u=[-1~0~-1~-2]$ satisfying the conditions in Lemma \ref{lem:lemma1_4}. Therefore, the repair coefficients vector in \eqref{eq:CoefVec} is given by
\[
{\bf y}' \erasurematrix^{\beta_{1}} = [0~~0~-1~-2].
\]
Consequently, 
\[
X_{1,1} =  - X_{1,3} - 2X_{1,4}.
\]
The remaining failure $X_{1,2}$ can also be repaired by the same procedure. 
\end{example}

\subsection{Inter-Rack Repair}

Communications across racks in a multi-rack storage network are in general  more expensive. Consider the extreme case where each rack physically represents a data center, each of which is geographically distant from each other. In this case, data transmission across long distance is clearly more expensive than transmission within each rack. Therefore, it is often desirable to design codes such that more repairs can be done locally within racks. 
However, in some rare cases (e.g., burst failure within a rack), nodes failure cannot be repaired locally.  For example, this may occur when  node $X_{m,j}$ fails and for all 
$\beta  \in \Omega(\Ha , j)$, there is at least another node $X_{m,k}$ for $k\in\beta$ which also fails. 
%
When intra-rack repair fails, inter-rack repair can be done. The idea is described below. 
%

Let 
$
\h = (h_{1}, \ldots, h_{N})
$, 
$
\k = (k_{1}, \ldots, k_{N})
$
and 
$
\g=(g_{1}, \ldots, g_{M}) 
$
be respectively vectors spanned by the rows of the matrices $\Ha$, $\Hb$ and  $\solG$. Then, it can be verified directly from \eqref{eq:IRP} and \eqref{eq:ARP} that 
\begin{align}
\h X_{m,*}^{\top} &= 0 \\
\k X_{m,*}^{\top} & = -g_{m}^{-1} \sum_{i\in\tau\setminus\{ m\}} \k \:g_{i}X_{i,*}^{\top} 
\end{align}
where $\tau = \{i \in \M:\: g_{i} \neq 0 \}$ and is assumed to contain $m$.
Suppose 
$\beta = \lambda(\h + \k)$
and $j \in\beta$. Then 
\begin{align}
(h_{j} +k_{j}) X_{m,j}
=  - \sum_{i\in \beta \setminus \{j\}} (h_{i}+k_{i}) X_{m,i} -g_{m}^{-1} \sum_{i\in \tau\setminus \{m\}} \k g_{i}X_{i,*}^{\top}  \label{eq7}
\end{align}
can be used to recover $X_{m,j}$. Equation \eqref{eq7} consequently defines the across rack repairs. To be more precise, in order to repair the failed node $X_{m,j}$, one would need 1) code symbols $X_{m,i}$ from the failing rack $m$ for $i \in \beta \setminus \{j\}$, and 2) code symbols $X_{i,\ell}$ from rack $i$ for $i\in \tau\setminus \{m\}$ and $\ell \in \{i \in \N:\: k_{i} \neq 0 \}$. In other words, to repair a failed node $X_{m,j}$ a group of helper racks $\tau$ are identified by parity matrix $\solG$. Also, a group of helper nodes in each helper rack is identified by parity matrix $\Hb$. The helper nodes in each helper rack will send their content to the rack process unit. Each helper rack process unit calculates a linear combination of the helper nodes content and send it to the process unit of the failed rack $m$. The process unit of the failed rack $m$ calculates the sum of this information received from helper racks. A group of survived nodes from the failed rack which are specified by $\Ha$ and $\Hb$ send their content to the rack process unit. The process unit then calculates a linear combination of the information from these nodes and adds it to the information from the helper racks. This results in the information content of the failed node $X_{m,j}$.

\begin{theorem}\label{thm:2}
Suppose node $j$ fails in rack  $m=1$. Let 
 $\pattern_{j}$ be the index set for all failed nodes (hence,  $j \in \pattern_{j}$). 
If $(\beta_{j} , \mu_{j} , \r_{j}, \tau )$ satisfies the following criteria, 
\begin{enumerate}

\item \label{item:ARRPcr1}
$\r_{j} \in \langle \Hb\rangle$ 

\item \label{item:ARRPcr2}
$\mu_{j} = \lambda(\r_{j})$  (i.e., $\mu_{j} = \left\{ n  \in \{1, \ldots, N \}: r_{j,n} \neq 0 \right\}$)

\item  \label{item:ARRPcr3}
$\beta_{j} \in \Omega(\Ha, \r_{j} , j)$, and

\item  \label{item:ARRPcr4}
$\beta_{j} \cap \pattern_{j} = \emptyset$,

\item  \label{item:ARRPcr5}
$\tau  \subseteq \{1, \ldots, M \} \in \Omega(\solG, 1)$

\end{enumerate}
then there exists $c_{j,n}$ for $n\in \beta_{j}$ and $d_{j,m,s}$ for $m\in\tau, s\in\mu_{j}$ such that 
\begin{align}\label{eq:thm2}
X_{1,j} = \sum_{m\in\tau} \left( 
                                        \sum_{s\in \mu_{j}} d_{j,m,s} X_{m,s}
  			 	    \right)
	     + \sum_{n \in \beta_{j}} c_{j,n} X_{1,n}.
\end{align}
\end{theorem}

The proof of Theorem \ref{thm:2} is given in Appendix \ref{append:A2}.

\begin{remark}
The interpretation of the theorem is as follows: 
The support of  $\r_j \in \langle \Hb\rangle$ corresponds to index of nodes in the ``helper racks''. Clearly, the smaller is the support the better, in order to minimise transmission cost. However, we would also point out that the  transmission costs required to transmit across racks does not depend on the support size of  
$\r_j $. More precisely, for each helper rack, only the sum 
$ \sum_{s\in \mu_{j}} d_{j,m,s} X_{m,s}$ is required to be transmitted, instead of specific individual $X_{m,s}$.
 On the other hand, the set  $\beta_j$ denotes the set of nodes which can be used to repair $X_{1,j}$. Consequently, $\beta_j$ and $\pattern_{j}$ (index set for the failed nodes in rack $1$) must be disjoint. Finally,  $\tau $ is the index set of the helper racks.
Note also that Theorem \ref{thm:2} reduces to Theorem \ref{thm1} if $\tau = \mu_{j}= \emptyset$ and $\r_{j}$ is the zero vector. 

\end{remark}
 
\begin{remark}
As a consequence of Theorem \ref{thm:2}, The processing unit in rack $m$ where $m\in \tau$, will retrieve $|\mu_{j}|$ symbols. The processing unit of rack 1, will need to retrieve $|\beta_{j}|$ symbols within the rack. Also, one symbol transmission is needed for the processing unit of rack 1 to send the recovered symbol back to the failed storage node $X_{1,j}$. Finally, each helper rack indexed in $\tau$ will transmit 1 symbol to the processing unit of rack 1. Summing up all these transmissions, there are in total 1) 
$
|\beta_{j}| + |\mu_{j}||\tau| + 1
$
symbol transmission within racks, and 2) $|\tau|$ symbol transmissions across racks. 
\end{remark}
 
\def\support{{\bf Supp}}

\begin{example} \label{examp:ARRP}
Consider a rack model storage network with $M=5$ racks, each of which contains $N=8$ storage nodes. Suppose the parity check matrices are as follows:
\[
\H = 
\left[
\begin{array}{cccccccc}
1 & 1 & 1 & 0 & 1 & 0 & 0 & 0 \\
1 & 1 & 0 & 1 & 0 & 1 & 0 & 0 \\
0 & 1 & 1 & 1 & 0 & 0 & 1 & 0 \\
1 & 0 & 1 & 1 & 0 & 0 & 0 & 1 \\
\end{array}
\right],
\]
\[
\Hb = 
\left[
\begin{array}{cccccccc}
1 & 1 & 0 & 1 & 1 & 0 & 0 & 1 \\
0 & 1 & 1 & 0 & 1 & 0 & 1 & 1 \\
\end{array}
\right],
\]
and
\[
\solG = 
\left[
\begin{array}{ccccc}
1 & 1 & 1 & 1 & 0 \\
0 & 1 & 1 & 1 & 1 \\
1 & 1 & 0 & 1 & 1 \\
\end{array}
\right].
\]

Note that the code is over $GF(2)$. Suppose node $X_{1,1}$ fails. Then  by Definition \ref{def:RepairGroup},  it can be verified that 
\begin{multline*}
\Omega(\H,1) = \Big\{\{3,4,8\}, \{2,7,8\}, \{2,4,6\}, \{3,6,7\}, \{2,3,5\}, \\
\{4,5,7\}, \{5,6,8\}, \{2,3,4,5,6,7,8\}\Big\}. 
\end{multline*}
In particular, any subset of nodes in rack $1$ indexed by $\Omega(\H,1)$ can be  used to repair $X_{1,1}$.

Now, suppose that the following nodes   $ \{X_{1,1}, X_{1,2}, X_{1,4}, X_{1,6}\}$ failed in  rack 1 (i.e., $\pattern = \{1,2,4,6 \}$). In this case, $X_{1,1}$ cannot be repaired via intra-rack repair  since there exist no intra-rack repair group $\beta_1 \in \Omega(\H,1)$ such that $\beta_1 \cap \gamma=\emptyset$. Therefore, inter-rack repair is needed. Let $\r_1=\left[0 \quad 1 \quad 1 \quad 0 \quad 1 \quad 0 \quad 1 \quad 1\right]$ and hence $\mu_1= \{2, 3, 5, 7, 8\}$  following from criterion \ref{item:ARRPcr1} and \ref{item:ARRPcr2}, respectively. Moreover, let $\beta_1 = \{7,8\}$ and $\tau = \{2,3,4\}$ following from criteria \ref{item:ARRPcr3}--\ref{item:ARRPcr5}, respectively. Note that, $\tau$, $\mu_1$, and $\beta_1$ indicate the group of helper racks, the group of helper nodes in the helper racks, and a repair group in rack 1 which will participate in repairing the failed node $X_{1,1}$.
 
Choose ${\bf y} = [1 \quad 0 \quad 0 \quad 0]$, ${\bf y}'=\left[0 \quad 0 \quad 0 \quad 0 \quad 0 \quad 0 \quad 1 \quad 1\right]$, ${\bf z} = [1 ~ 0 ~ 0]$, ${\bf z}' = [0 ~ 1 ~ 1~ 1~ 0]$, and $a=1$ such that 
\begin{align} \label{eq:e1}
{\bf e}_{1} &= {\bf y} \Ha + {\bf y}' \erasurematrix^{\beta_{1}} + a  \r_{1}, \\ \label{eq:f1}
{\bf f}_1 & = {\bf z}  \solG  + {\bf z}' \erasurematrix^{\tau},
\end{align}
are satisfied where \eqref{eq:e1} and \eqref{eq:f1} follow from Lemma \ref{lem:lemma1_4}, and ${\bf f}_1$ is defined in \eqref{eq:fVector}. Following from \eqref{eq:cVector} and \eqref{eq:dmatrix}, we have
\[
[c_{1,1} , \ldots, c_{1,8}] = 
\left[
\begin{array}{cccccccc}
0 & 0 & 0 & 0 & 0 & 0 & 1 & 1
\end{array}
\right],
\]
and
\[
\left[
\begin{array}{ccc}
d_{1,1,1} & \cdots & d_{1,1,8} \\
\vdots & \ddots &\vdots \\
d_{1,5,1} & \cdots, & d_{1,5,8}  \\
\end{array}
\right] =
\left[
\begin{array}{cccccccc}
0 & 0 & 0 & 0 & 0 & 0 & 0 & 0 \\
0 & 1& 1 & 0 & 1 & 0 & 1 & 1 \\
0 & 1& 1 & 0 & 1 & 0 & 1 & 1 \\
0 & 1& 1 & 0 & 1 & 0 & 1 & 1 \\
0 & 0 & 0 & 0 & 0 & 0 & 0 & 0
\end{array}
\right]. 
\]
Now, the failed node $X_{1,1}$ can be recovered by \eqref{eq:thm2} such that
\begin{align} \nonumber
X_{1,1} & = \sum_{m\in\{2,3,4\}} \left( 
                                        \sum_{s\in \{2,3,5,7,8\}} d_{1,m,s} X_{m,s}
  			 	    \right)   + \sum_{n \in \{7,8\}} c_{1,n} X_{1,n}.
\end{align}

In this example a total number of 18 symbol transmissions within the racks and 3 symbol transmissions across the racks are needed to repair the failed node $X_{1,1}$.  
\end{example}

When we need to repair multiple failed nodes (say $|\pattern_{j}|$ of them) via inter-rack repair, the cost is not simply $|\pattern_j|$ times:  First, it is possible that a transmission from inter-rack can be used to repair for more than one node. Second, once inter-rack repair has been achieved, nodes which are previously not repairable via intra-rack repair may become repairable. 
To be more precise, suppose nodes $(X_{1, j} , j \in \pattern)$ fail where $|\pattern| \ge  \dist(\Ha) $. In this case, nodes failure may not be able to be recovered merely via intra-rack repairs. Let $\alpha\subseteq \pattern$ be of size $\dist(\Ha)  - 1$. In that case, 
in the worst case scenario, one can at least aim to recover variables $X_{1, j}$ for $j \in \pattern\setminus\alpha$ via inter-rack repair first. Once this is achieved, the remain nodes 
$(X_{1, j},  j \in  \alpha)$ can be recovered via intra-rack repair. 
Following the idea, the following theorem gives an upper bound on the repair transmission cost.

\begin{theorem}[Upper bound on transmission costs] \label{thm:cost}
Let $\pattern$ be the index set of all failed nodes in rack $1$. 
Suppose that 1) for any $j \in \pattern\setminus \alpha$, there exists $(\beta_{j}, \mu_{j} , \r_{j}, \tau)$ satisfying the criteria in Theorem \ref{thmLPbd} where $\pattern_{j} \triangleq \pattern$ and 2) for any $j \in \alpha$, there exists $\beta_j$ satisfying the criteria in Theorem \ref{thm1} where $\pattern_{j} \triangleq \alpha$.
%
 Then the required total transmissions within a rack  $\theta_{intra}$ and across  racks $\theta_{inter}$ are respectively upper bounded by  
\begin{align}
\theta_{intra} &\leq |\tau| \left| \bigcup_{j\in\pattern\setminus\alpha}\mu_{j}\right| + 
 \left|\bigcup_{j \in \pattern }\beta_{j} \right|   + |\pattern|  \label{thm3:eqa}\\
\theta_{inter} & \leq |\tau| \dim \langle \r_j , j\in \pattern\setminus\alpha \rangle. \label{thm3:eqb}
\end{align}
\end{theorem} 

The proof of Theorem \ref{thm:cost} is given in Appendix \ref{append:A3}.

\subsection{Code Rate} 
In this subsection we derive the rate of our proposed code for multi-rack storage networks. The code rate will later be employed to establish the upper bound of the code size in Section \ref{Sec:V}. The following theorem gives the rate of the multi-rack storage code.

\begin{theorem} \label{thm:CodeRate}
Let $\Ha$, $\Hb$, and $\solG$ be respectively $S_{1} \times N$, $S_{2} \times N$, and $L \times M$ matrices. Then the rate of the multi-rack storage code 
$(\Ha, \Hb, \solG)$ is lower bounded by 
\begin{align} \label{eq:RackCodeRate1}
R \geq \frac{MN - MS_{1} - LS_{2}}{MN}.
\end{align}
Equality holds if rows in $\Ha$ and $\Hb$ are linearly independent, and $\solG$ is a full rank matrix.
\end{theorem}

The proof of Theorem \ref{thm:CodeRate} is given in Appendix \ref{append:A4}.

\section{Bounds} \label{Sec:V}

In this section, we first derive the relations between the code rate, code size, and the size of the parity check matrices. We show that under some constraints, maximising the code rate is equivalent to maximising the size of the code which in turn can determine the optimum size of the parity check matrices. We define the multi-rack storage network parameters such as intra- and inter-rack resilience and locality in order to establish a \emph{Linear Programming (LP)} problem to maximise the size of the multi-rack storage code. Then, the symmetries in the problem will be used to significantly reduce the complexity of the LP problem.

\subsection{Linear Programming Bound} \label{subsec:IVA}

In the previous section, we introduced a class of storage codes called multi-rack storage codes and explained how to repair nodes failure via intra-rack or inter-rack repairs. In this section, we will develop   bounds for this class of codes. 
Recall our code construction and definition. We will notice the following:
\begin{enumerate}
\item 
The intra-rack parity check matrix $\Ha$ or more precisely the support of the dual codewords spanned by the rows of $\Ha$ determines how failed nodes can be repaired locally within a rack. Alternatively,  the dual codewords spanned by intra-rack parity check matrix $\Ha$ and inter-rack parity check matrix $\Hb$ together defines the inter-rack repair process.  

\item 
The helper-rack parity check matrix $\solG$ specifies which racks can be used in the inter-rack repair process.  Naturally, one would prefer to involve only a small number of racks to minimize the inter-rack transmission cost. 

\end{enumerate}

Assuming without loss of generality that all  rows of   $\Ha$ and  $\Hb$ are  independent and that $\solG$ is full rank,
Theorem \ref{thm:CodeRate} shows that 
\begin{align}\label{eq24}
R(\Lambda_{\C}) = \frac{MN - MS_{1} - LS_{2}}{MN}, 
\end{align}
or equivalently, 
\begin{align}\label{eq25}
R(\Lambda_{\C}) = \frac{ N  - S_{1} - S_{2} }{N}+\frac{S_{2}(M-L)}{MN}.
\end{align}
The rate of the multi-rack storage code is essentially determined by the size of the matrices $\Ha, \Hb$ and $\solG$. 

Understanding their roles, we can immediately recognise that one can separately design $\solG$ and $(\Ha, \Hb)$. The design of $\solG$ will affect the number of helper racks needed in inter-rack repair.  In fact, it is  very similar to the design of locally repairable codes. The idea is to design a linear code (specified by the parity check matrix $\solG$) such that for any $m \in \M$, there are dual codewords $\g$ with a small support containing $m$. If the racks are geographically separated and connected to a network, the design of $\solG$ may take into account the network topology and the costs of the transmission links. For example, an algorithm for the design of linear binary locally repairable codes over a network can be found in \cite{LRCTopo}. There are also previous works including our own work in \cite{LPLR} and \cite{RLLC} which discuss the design and bounds for locally repairable codes.
On the other hand, the design of $\Ha$ and $\Hb$ will affect the code's ability in intra-rack and inter-rack repair. The focus of the remaining paper is on understanding the fundamental limits of the best design of these two matrices.

Separating the design of $(\Ha, \Hb)$ from $\solG$,  we can simply consider a   simple special case where  there are only two racks (i.e., $M = 2$) and that 
$$
\solG = [1,-1].
$$
Assume without loss of generality, we may characterise our multi-rack storage code via the following parity-check equations:
\begin{align} 
\Ha \X_{1,*}^{\top}  & = \zeroz  \\
\Ha \X_{2,*}^{\top}  & = \zeroz  \\
\Hb \X_{1,*}^{\top} & = \Hb \X_{2,*}^{\top}.
\end{align}
As such, we will simply refer a multi-rack storage code as $(\Ha,\Hb)$.

\begin{definition}\label{def2}
We call a multi-rack storage code $(\Ha,\Hb)$ as a  
$(\delta_1, \Gamma_{1}, r_1, \delta_2, \Gamma_{2}, r_2, a)$ linear multi-rack storage code  if it satisfies the following conditions:
\begin{enumerate}
\item (Intra-rack resilience)  
Any $\delta_1$ node failures in a rack can be repaired via intra-rack repair;

\item (Intra-rack  locality) 
For any $\Gamma_{1} + 1$ node failure pattern in a rack,  each node can be repaired via intra-rack repair, involving at  most $r_{1}$ surviving nodes;

\item (Inter-rack resilience) Any $\delta_2$ node failures in a rack can be repaired via inter-rack repair.

\item (Inter-rack locality) For any $\Gamma_{2} + 1$ node failure pattern in a rack,  each node can be repaired via inter-rack repair such that involving $i$) at most  $r_{2}$ surviving nodes in the failing rack and  $ii$) at most  $a$ nodes from each   helper rack.

\end{enumerate}
\end{definition}

\def\f{{\bf f}} 
\def\g{{\bf g}} 

The definition for $(\delta_1, r_1, \Gamma_{1}, r_2, \Gamma_{2}, a,  \delta_2)$ linear multi-rack storage code  
can be made more precise via the use of support enumerators, to be described as follows.

To simplify our notation, we may use $\bx$ and $\by$ instead of $X_{1,*}^{\top}$ and $X_{2,*}^{\top}$.  
Let 
\begin{align} \label{eq:TwoRackCode}
\C = 
\{
(\bx, \by) : \Ha \bx = \Ha \by = {\bf 0}, \: \Hb \bx = \Hb \by
\}.
\end{align}
We call $\C$ the codebook. Clearly, the design of $\C$ and the design of $(\Ha,\Hb)$ are equivalent.

\begin{definition}[Support]
Consider any codeword $(\bx, \by) \in \C$ where 
$\bx=(x_{1},\ldots, x_{N})$ and $\by=(y_{1},\ldots, y_{N})$. Its "support" $\lambda(\bx,\by)$ is a tuple $(\w , \s)$ such that $\w = (w_1, \ldots, w_N)$ and $\s = (s_1, \ldots, s_N)$, where
\begin{align}
w_{i}  = 
\begin{cases}
1 & \text{ if } x_{i} \neq 0 \\
0 & \text{ if } x_{i} = 0  
\end{cases} \label{eq12}\\
s_{i}  = 
\begin{cases}
1 & \text{ if } y_{i} \neq 0 \\
0 & \text{ if } y_{i} = 0,
\end{cases}\label{eq13}
\end{align}
for all $i=1, \ldots, N$. For notation simplicity, we will simply denote that 
$$
\lambda(\bx,\by) = (\w,\s).
$$
\end{definition}

\begin{remark}
While $\w$ and $\s$ are subsets of $\N$, it is sometimes simpler and more practical to represent them as vectors, as in \eqref{eq12} and \eqref{eq13}.
\end{remark}

\begin{definition}[Support enumerator]
The enumerator function of the code $\Lambda_\C(\w ,\s)$ is defined as
\begin{align}
\Lambda_\C(\w,\s)
\triangleq \left|\left\{(\bx,\by) \in \C : \lambda(\bx,\by) = (\w,\s)\right\}\right|
\end{align}
for all $\w, \s  \subseteq \N$.
\end{definition}
   

The below theorem gives properties of a multi-rack storage code. As we shall see, these properties will form constraints in our linear programming bound.

\begin{theorem}\label{thm5}
For any  $(\delta_1, \Gamma_{1}, r_1, \delta_2, \Gamma_{2}, r_2, a)$ multi-rack storage code $\C$,   
the support enumerators of $\C$ and its dual $\Cd$ satisfy the following properties:
\begin{enumerate}
\item \label{item:RackModelCodeProperty1}
{\bf{\emph{Dual codeword support enumerator:}}}
\begin{align}
%
\Lambda_{\Cd}(\w ,\s) 
= \frac{1}{|\C|} \sum_{\w' , \s' \subseteq \N} \Lambda_\C(\w', \s')   \prod_{j\in \N}  \kappa_q (w'_j ,w_j)\kappa_q (s'_j,s_j) \label{eqThm:a}
\end{align}
where
\[
\kappa_{q}\left(u,v\right)=
\begin{cases}
 1 & \text{ if } v=0 \\ 
 q-1 & \text{ if } u=0  \text{ and }  v=1 \\ 
 -1 & \text{ otherwise. } 
\end{cases}  
\]

\item \label{item:RackModelCodeProperty2}
{\bf{\emph{Symmetry:}}} 
For all $\w,\s \subseteq \N$,
\begin{align} \label{eqThm:b}
\Lambda_\C(\w, \s) = \Lambda_\C( \s,\w).  
\end{align}

\item \label{item:RackModelCodeProperty3}
{\bf{\emph{Intra-rack  resilience:}}}
For all $\w \subseteq \N$ such that   $1 \le |\w| \le \delta_1$
\begin{align}\label{eqThm:c}
\Lambda_\C(\w, \s) = 0
\end{align}

\item \label{item:RackModelCodeProperty5}
{\bf{\emph{Intra-rack locality:}}}
For any 
$
(i,\gamma) \in \Phi(\Gamma_{1})\triangleq \{
(i,\gamma) : \: i \in \N, i \not \in \gamma \text{ and } |\gamma| = \Gamma_{1}\}
$, 
\begin{align}\label{eqThm:d}
\displaystyle 
\sum_{\w \in \Theta_{1}(i,\gamma,r_1)} 
\Lambda_{\Cd}(\w,\emptyset) \geq (q-1), 
\end{align}
where 
$
\Theta_{1}(i,\gamma,r_{1}) 
\triangleq 
\{ 
\w  : \: i\in\w, \w \cap  \gamma  = \emptyset  \text{ and } |\w| \le r_{1}+1
\}.
$

\item \label{item:RackModelCodeProperty4}
{\bf{\emph{Inter-rack resilience:}}}
For all $\w \subseteq \N$ such that   $1 \le |\w| \le \delta_2$
\begin{align}\label{eqThm:e}
\Lambda_\C(\w, \emptyset)  = 0
\end{align}

\item \label{item:RackModelCodeProperty6}
{\bf{\emph{Inter-rack locality:}}} 
For any 
$
(i,\gamma) \in \Phi(\Gamma_{2})\triangleq \{
(i,\gamma) : \: i \in \N, i \not \in \gamma \text{ and } |\gamma| = \Gamma_{2}\}
$, 
\begin{align}\label{eqThm:f}
\displaystyle 
\sum_{ (\w,\s) \in  \Theta_{2}(i,\gamma, r_{2} , a) } 
\Lambda_{\Cd}(\w,\s) \geq (q-1),  
\end{align}
where
$
\Theta_{2}(i,\gamma, r_{2} , a)
 \triangleq 
\{ 
(\w,\s)  : \: i\in \w, \w \cap  \gamma  = \emptyset , 
|\w| \le r_{2}+1 \text{ and } |\s | \le a
\}.
$
\end{enumerate}
\end{theorem}

The proof of Theorem \ref{thm5} is given in Appendix \ref{append:A}.

\begin{lemma} Consider a multi-rack storage code $(\Ha, \Hb , \solG)$. Suppose the dimensions of the matrices $\Ha, \Hb , \solG$ are respectively $S_{1} \times N$, $S_{2} \times N$, and $L\times M$. Then 
\begin{align}\label{eq62}
N-S_{1}-S_{2} = \log_{q} \sum_{\w\subseteq\N}\Lambda_{\C}(\w,\emptyset) 
\end{align}
and
\begin{align}\label{eq63}
S_{2} =  \log_{q} \sum_{\w,\s\subseteq\N}\Lambda_{\C}(\w,\s)  - 2 \log_{q} \sum_{\w\subseteq\N}\Lambda_{\C}(\w,\emptyset).
\end{align}
Hence, the rate of the code is 
\begin{multline} \label{eq:CodeRateCodeSize}
R(\Lambda_{\C}) = \frac{  \log_{q} \sum_{\w\subseteq\N}\Lambda_{\C}(\w,\emptyset)}{ N} 
\\
+ 
\frac{M-L}{MN} \left( 
\log_{q} \sum_{\w,\s\subseteq\N}\Lambda_{\C}(\w,\s)  - 2 \log_{q} \sum_{\w\subseteq\N}\Lambda_{\C}(\w,\emptyset)
\right).
\end{multline}
\end{lemma}
\begin{proof}
First, it is clear that 
$
\sum_{\w\subseteq\N}\Lambda_{\C}(\w,\emptyset) 
$
is equal to the size of the following set
\[
|\{
(\bx, {\bf 0}) : \Ha \bx =  {\bf 0}, \: \Hb \bx = {\bf 0}
\}
|
\]
As the dimensions of $\bx, \Ha$ and $\Hb$ are respectively 
$N \times 1$, $S_{1} \times N$ and $S_{2} \times N$, the size of the set is 
obviously $ q^{N-S_{1}-S_{2}}$ leading to \eqref{eq62}

Since $\Ha$ and $\Hb$ are respectively $S_1 \times N$ and $S_2 \times N$ parity matrices of the code $\C$ in \eqref{eq:TwoRackCode}, $i$) the total number of parity equations is $2S_1+S_2$, and $ii$)  the length of codeword $(\bx,\by)$ is $2N$. Hence, 
the total number  of  codewords satisfying \eqref{eq:TwoRackCode}   is
\[
\sum_{\w,\s\subseteq\N}\Lambda_{\C}(\w,\s)=q^{2N-2S_1-S_2}.
\] 
Together with \eqref{eq62}, we have   \eqref{eq63} and \eqref{eq:CodeRateCodeSize}.
\end{proof}

\def\Oa{{O_{1}}}

\def\TOT{{T}} 

\begin{remark}
According to \eqref{eq:CodeRateCodeSize}, the rate of the storage code  is clearly nonlinear, with respect to the support enumerator $\Lambda_{\C}$. However, if we fix 
$$
\sum_{\w\subseteq\N}\Lambda_{\C}(\w,\emptyset) = \Oa,
$$
then maximising $R(\Lambda_{\C})$ is equivalent to maximising 
$\sum_{\w,\s\subseteq\N}\Lambda_{\C}(\w,\s)$.
\end{remark}

\begin{theorem}[Upper bound]\label{thmLPbd}
Consider fixed $N$, $M$ and $L$ and   a $(\delta_1, r_1, \Gamma_{1},  \delta_2, r_2, \Gamma_{2}, a)$ multi-rack storage code $\C$ such that 
$$
\sum_{\w\subseteq\N}\Lambda_{\C}(\w,\emptyset) = \Oa.
$$
Let $O^{*}$ be the maximum of the following 
linear programming  problem.

\medskip
\noindent
\underline{Linear Programming Problem (LP1)}
\begin{align}
\displaystyle
&\textbf{Maximize}  \quad \sum_{\w,\s \subseteq \N} A_{\w,\s}  \nonumber\\
&\textbf{subject to}  \nonumber\\
& \quad  
A_{\w,\s} \geq 0,  \quad \forall \w,\s \subseteq \N \tag{C1}\\
& \quad
A_{\w,\s}  = A_{\s,\w}  \tag{C2}\\
&\quad C_{\w,\s}  = \sum_{\w' , \s' \subseteq \N} A_{\w',\s'}   \prod_{j\in \N}  \kappa_q (w'_j ,w_j)\kappa_q (s'_j,s_j), 
 \quad \forall \w,\s \subseteq \N  \tag{C3}\label{eq:C3} \\ 
& \quad C_{\w, \s}  \geq 0,   \quad \forall \w,\s \subseteq \N \tag{C4} \label{eq:C4} \\ 
& \quad A_{\emptyset,\emptyset} = 1 &  \tag{C5}\label{eq:C5} \\ 
& \quad
A_{\w,\s} = 0,
\quad \forall 1 \le |\w| \le \delta_1 \tag{C6} \label{eq:C6}\\ 
& \quad
A_{\w,\emptyset}  = 0,
\quad \forall 1 \le |\w| \le \delta_2 \tag{C7}\\ \displaybreak[3]
& \quad
\sum_{\w \in \Theta_{1}(i,\gamma,r_1)} 
C_{\w,\emptyset} \geq (q-1)  \sum_{\w,\s} A_{\w,\s},
\quad \forall (i,\gamma) \in \Phi(\Gamma_{1}) \tag{C8} \\ 
& \quad
\sum_{ (\w,\s) \in  \Theta_{2}(i,\gamma, r_{2} , a) } 
C_{\w,\s} \geq (q-1)  \sum_{\w,\s} A_{\w,\s}, 
\quad \forall (i,\gamma) \in \Phi(\Gamma_{2}) \tag{C9} \label{eq:C9}\\
& \quad
 \sum_{\w \subseteq \N}A_{\w,\emptyset} = \Oa \tag{C10}
\end{align} 
Then $R(\Lambda_{\C})$  is upper bound by  
\[
\frac{\log_{q} \Oa}{N} + 
\frac{M-L}{MN} 
\left( 
\log_{q} O^{*}  - 2 \log_{q}\Oa 
\right).
\]
\end{theorem}

\begin{proof}[Proof of Theorem \ref{thmLPbd}]
We define
\begin{align*}
A_{\w,\s} & \triangleq \Lambda_\C(\w,\s) \\
C_{\w,\s} & \triangleq |\C| \Lambda_{\Cd}(\w,\s).
\end{align*}
As mentioned earlier, maximizing the code rate is equivalent to maximize the code size which is the objective function of the linear programming problem LP1. The first constraint of the optimization problem LP1 follows from the fact that the number of codewords are non-negative. The second constraint follows from the symmetry property of the code.  The constraint \eqref{eq:C3} follows from the dual code support enumerator property (MacWilliam's identity) in Theorem \ref{thm5}. Constraint \eqref{eq:C4} follows from the fact that the number of dual codewords are non-negative. Constriant \eqref{eq:C5} follows from the fact that there exists only one zero codeword in code $\C$. Constraints \eqref{eq:C6}--\eqref{eq:C9} follow from the properties \ref{item:RackModelCodeProperty3}--\ref{item:RackModelCodeProperty6} in Theorem \ref{thm5}.
\end{proof}

\begin{remark}
Strictly speaking, to optimise $R(\Lambda_{\C})$, one also needs to optimise the choice of $O_{1}$, which is generally unknown. However, as $O_{1} \in  \{q^{i}, \: i=0, \ldots, N \}$. Hence, 
we have 
\[
R(\Lambda_{\C}) \le \max_{ i=0, \ldots, N}
\frac{i}{N} + 
\frac{M-L}{MN} 
\left( 
\log_{q} O^{*}(i)  - 2 i 
\right)
\]
where $O^{*}(i)$ is   the maximum of (LP1) when $\Oa = q^{i}$.
\end{remark}

\subsection{Bound Simplification via Symmetry}

The complexity (in terms of the number of variables and the number of constraints) of the linear programming problem LP1 in Theorem \ref{thmLPbd} will increase exponentially with the number of storage nodes $N$ in each rack.  
However, if we notice the LP1 carefully, we can observe that the problem itself has much symmetries in it such that exploiting this inherent symmetry can significantly reduce the problem complexity.  

Let $S_\N$  be the symmetric group on $\N$, whose elements are all the permutations of the elements in $\N$ which are treated as bijective functions from the set of symbols to itself. Clearly, $|S_\N| = N!$.
%
The variables in the optimization problem LP1 are 
$$
(A_{\w,\s}, C_{\w,\s}, \: \w,\s\subseteq \N).
$$
Let $\sigma$ be a permutation on $\N$ such that $\sigma \in S_\N$. For each $\w \subseteq \N$, we extend the mapping $\sigma$ by defining 
$$
\sigma(\w) \triangleq \{  \sigma(i) : \: i \in \w \}.
$$
Due to the symmetries, we have the following proposition. 
 
 \def\qqq{}
\begin{proposition} \label{prop1}
Suppose $(\qqq a_{\w,\s}, c_{\w,\s} :\: \w,\s \subseteq \N)$ satisfies all the constraints in the linear programming problem LP1 in Theorem \ref{thmLPbd}. For any $\sigma \in S_\N$,  let
\begin{align}
a_{\w,\s}^{(\sigma)} & = a_{\sigma(\w),\sigma(\s)} \\
c_{\w,\s}^{(\sigma)} & = c_{\sigma(\w),\sigma(\s)}. 
\end{align}
Then 
$(\qqq a_{\w,\s}^{(\sigma)}, c_{\w,\s}^{(\sigma)} :\: \w,\s \subseteq \N)$ also satisfies the constraints in LP1, with the same values in the objective function. In other words,
\[
\sum_{\w,\s\subseteq \N}a_{\w,\s}^{(\sigma)} = \sum_{\w,\s\subseteq \N}a_{\w,\s}.
\]
\end{proposition}

\begin{proof}
The proposition  follows directly from the symmetry in the constraint and optimising function.
\end{proof}

As the feasible region in the linear programming problem (LP1) is convex, we have the following corollary.

\begin{corollary}\label{corollary1}
Let 
\begin{align}  
a^*_{\w,\s} & = \frac{1}{|S_\N|} \sum_{\sigma \in S_\N} a_{\w,\s}^{(\sigma)} \label{eq:AvgSol} \\
c^*_{\w,\s} & = \frac{1}{|S_\N|} \sum_{\sigma \in S_\N} c_{\w,\s}^{(\sigma)} 
\end{align}
Then 
$(\qqq a_{\w,\s}^{*}, c_{\w,\s}^{*} :\: \w,\s \subseteq \N)$ also satisfies the constraints in (LP1) and  
\begin{align} \label{eq:SumAvg}
\sum_{\w,\s \subseteq \N} a^*_{\w,\s} = \sum_{\w,\s \subseteq \N} a_{\w,\s}.
\end{align}
\end{corollary}

\begin{proof}
From Proposition \ref{prop1}, for any feasible solution $(a_{\w,\s}, c_{\w,\s} : \: \w,\s \subseteq \N)$, there exist $|S_\N|$ other feasible solution $(a^{(\sigma)}_{\w,\s}, c^{(\sigma)}_{\w,\s} : \: \w,\s \subseteq \N, \sigma \in S_\N)$. Since $(a^*_{\w,\s},c^*_{\w,\s})$ is the convex linear combination of all these  feasible solutions,  it also satisfies the constraints of LP1. From \eqref{eq:AvgSol}
\[
\sum_{\w,\s \subseteq \N} a^*_{\w,\s} = \frac{1}{|S_\N|} \sum_{\w,\s \subseteq \N} \sum_{\sigma \in S_\N} a_{\w,\s}^{(\sigma)}.
\]
Moreover, from Proposition \ref{prop1}
\[
\sum_{\w,\s \subseteq \N} \sum_{\sigma \in S_\N} a_{\w,\s}^{(\sigma)} = |S_\N| \sum_{\w,\s \subseteq \N} a_{\w,\s}.
\]
The corollary then follows.
\end{proof}

By Corollary \ref{corollary1},  it is sufficient to consider only "symmetric" feasible solutions of the form 
$
(\qqq a^*_{\w,\s}, c^*_{\w,\s} : \w,\s \subseteq \N).
$
in LP1. In other words,  one can impose additional symmetric constraint to LP1 without affecting the value of the objective function.

One important benefit for considering $
(\qqq a^*_{\w,\s}, c^*_{\w,\s} : \w,\s \subseteq \N).
$ is that many terms in the linear programming bound can become alike (and hence can be grouped together).

\begin{proposition}[Grouping alike terms] \label{prop2}
Suppose $\w,\s,\w',\s' \subseteq \N$ such that
\begin{align}
|\w \setminus \s | & = |\w' \setminus \s' | \label{eq35} \\
|\w \cap \s | & = |\w' \cap \s'  | \\
|\s \setminus \w | & = |\s' \setminus \w' | \label{eq37}.
\end{align}
Then $a^*_{\w,\s}  = a^*_{\w', \s'}$ and $c^*_{\w,\s} = c^*_{\w', \s'} $.
\end{proposition}

The proof of Proposition \ref{prop2} is given in Appendix \ref{append:A5}.

Due to  Proposition \ref{prop2},  we can impose the following additional constraint on (LP1) 
\begin{align}
A_{\w,\s} &= A_{\w',\s'} \\
C_{\w,\s} &= C_{\w',\s'} 
\end{align}
for all $\w, \s$ satisfying \eqref{eq35} - \eqref{eq37}. 
As many of these variables are now the same, one can greatly reduce the number of variables in (LP1). 

\begin{theorem}[Simplified LP Bound]\label{thm4}
The maximum in (LP1) is the same as the maximum of the following linear programming problem:


\noindent
\underline{Reduced Linear Programming Problem (LP2)} \label{LP2}
\begin{align}
\displaystyle
&\textbf{maximize}  \quad \sum_{d,e,f }  \binom{N}{d,e,f} X_{d,e,f } \nonumber\\
&\textbf{subject to}  \nonumber\\
& \quad  
X_{d,e,f} \geq 0,  \hspace{3.8cm} \quad \forall d,e,f   \tag{D1}\\
& \quad
X_{d,e,f}  = X_{f,e,d},  
\hspace{2.85cm} \quad \forall d,e,f  \tag{D2}\\   
& \quad
Y_{d,e,f} = \sum_{d',e',f'} \Delta_{1}(d,e,f,d',e',f') X_{d',e',f'}, \quad  
 \forall d,e,f
\tag{D3}\\ 
& \quad Y_{d,e,f}  \geq 0, \hspace{3.75cm}  \quad \forall d,e,f \tag{D4} \\ \displaybreak[3]
& \quad X_{0,0,0} = 1 &  \tag{D5}\\ 
& \quad
X_{d,e,f} = 0,
\hspace{1.45cm} \quad \forall 1 \le d + e \le \delta_1 \tag{D6}\\ 
& \quad
X_{d,0,0}  = 0,
\hspace{2.8cm} \quad \forall 1 \le d\le \delta_2 \tag{D7}\\
& \quad
\sum_{d=2}^{r_{1}+1} \Delta_{2}(d) Y_{d,0,0} \ge (q-1) \sum_{d,e,f} \binom{N}{d,e,f} X_{d,e,f},
\tag{D8} \\
& \quad
\sum_{d+e\leq r_2 +1,e+f\leq a} \Delta_{3}(d,e,f) Y_{d,e,f} \ge (q-1) \sum_{d,e,f}  \binom{N}{d,e,f} X_{d,e,f}, \tag{D9}\\
& \quad
 \sum_{d} X_{d,0,0}  = \Oa 
 \tag{D10}
%
\end{align} 
where $ \Delta_{1}(d,e,f,d',e',f') $, $ \Delta_{2}(d)$ and $\Delta_{3}(d,e,f) $ are respectively defined as in  \eqref{eqDelta1}, \eqref{eqDelta2} and \eqref{eqDelta3}.
Here,  $ (d,e,f)$  are tuples such that $d,e,f$ are nonnegative integers with  a total sum no more than $N$.
\end{theorem}

The proof of Theorem \ref{thm4} is given in Appendix \ref{append:B}.

\begin{remark}
In (LP1), the  number of variables and constraints grows exponentially with $N$ -- the number of nodes in a rack. Such exponential growth makes  (LP1) practically infeasible to solve for moderate $N$. However, via reduction by symmetry,  the number of variables in (LP2) has greatly reduced to  
$
2\binom{N+3}{3} + 1
$
while  the number of constraints  to  
$
\frac{7}{2} \binom{N+3}{3} + 5 + \delta_1 + \delta_2.
$
Clearly, the reduction is significant. 
\end{remark}

\begin{example}
We consider a very simple setup to show how some of the constraints in linear programming bound in (LP2) are calculated. Let, $N=8$ is the number of the storage nodes in each rack. As mentioned earlier in Section \ref{subsec:IVA} parity matrix $\solG$ can be design separately, hence the number of the racks $M$ will not affect the LP bound. Moreover, let $\Gamma_1=2$, $\delta_1=3$, $r_1=3$, $\Gamma_2=4$, $\delta_2=6$, $r_2=1$, and $a=3$. Then, we have $(0 \leq d+e+f \leq 8, 0 \leq d,e,f \leq 8)$ and $(0 \leq d^{\prime}+e^{\prime}+f^{\prime} \leq 8, 0 \leq d^{\prime},e^{\prime},f^{\prime} \leq 8)$. Here, we show how to calculate $\Delta_{1}(d,e,f,d',e',f')$, $\Delta_{2}(d)$, and $\Delta_{3}(d,e,f)$. Following \eqref{eqDelta2} and \eqref{eqDelta3}, respectively,
\begin{gather*}
\Delta_{2}(d) = \binom{5}{d-1}, \\
\Delta_{3}(d,e,f) = \binom{3}{e-1}\binom{3-e}{d}\binom{8-e-d}{f}+\binom{3}{e}\binom{3-e}{d-1}\binom{8-e-d}{f}.
\end{gather*}
In order to calculate $\Delta_{1}(d,e,f,d',e',f')$ from \eqref{eqDelta1}, we need to find all possible values of $(\zeta_{i,j},  1 \leq i,j \leq 4)$  for each fixed tuples of $(d,e,f)$ and all tuples of $(d',e',f')$ such that $d=\sum_j \zeta_{1,j}$, $e=\sum_j \zeta_{2,j}$, $f=\sum_j \zeta_{3,j}$, $d'=\sum_j \zeta_{j,1}$, $e'=\sum_j \zeta_{j,2}$, $f'=\sum_j \zeta_{j,3}$. Then $U(\zeta)$, $\sigma_1(\zeta)$, and $\sigma_2(\zeta)$ can be calculated from \eqref{eq:sigma1zeta}, \eqref{eq:sigma2zeta}, and \eqref{eq:Uzeta}, respectively. For instance, let $(d,e,f)=(5,2,1)$. Then for a tuple $(d',e',f')=(3,1,0)$, one possible combinations for $\zeta_{i,j}$ will be $(\zeta_{1,1}=2, \zeta_{1,2}=0, \zeta_{1,3}=0, \zeta_{1,4}=3, \zeta_{2,1}=1, \zeta_{2,2}=0, \zeta_{2,3}=0, \zeta_{2,4}=1, \zeta_{3,1}=0, \zeta_{3,2}=1, \zeta_{3,3}=0, \zeta_{3,4}=0, \zeta_{4,1}=0, \zeta_{4,2}=0, \zeta_{4,3}=0, \zeta_{4,4}=0)$. Therefore, 
\begin{gather*}
\sigma_1(\zeta)=6, ~~ \sigma_2(\zeta) = 4,\\
U(\zeta) = \binom{5}{2,0,0,3}\binom{2}{1,0,0,1}\binom{1}{0,1,0,0}\binom{0}{0,0,0,0}=20.
\end{gather*}
Another possible combination will be $(\zeta_{1,1}=1, \zeta_{1,2}=0, \zeta_{1,3}=0, \zeta_{1,4}=4, \zeta_{2,1}=1, \zeta_{2,2}=1, \zeta_{2,3}=0, \zeta_{2,4}=0, \zeta_{3,1}=1, \zeta_{3,2}=0, \zeta_{3,3}=0, \zeta_{3,4}=0, \zeta_{4,1}=0, \zeta_{4,2}=0, \zeta_{4,3}=0, \zeta_{4,4}=0)$. Therefore, 
\begin{gather*}
\sigma_1(\zeta) = 6, ~~ \sigma_2(\zeta) = 4,\\
U(\zeta) = \binom{5}{1,0,0,4}\binom{2}{1,1,0,0}\binom{1}{1,0,0,0}\binom{0}{0,0,0,0}=10.
\end{gather*}
As mentioned earlier, all values of $\zeta_{i,j}$ must be calculated (this can be done by a computer simulation code e.g. in MATLAB) for all $(d,e,f)$ and $(d',e',f')$ in order to calculate $\Delta_1(d,e,f,d',e',f')$. Figure \ref{fig:ZetaTable} illustrates of selecting the different combinations for $\zeta_{i,j}s$. The numbers in red and green are the two aforementioned different possible combinations for $\zeta_{i,j}s$. The combinations will be chosen such that the sum of the rows and columns satisfy the values of $(d,e,f)$ and $(d',e',f')$.

\begin{figure}[t]
\centering
\includegraphics[width=.5\textwidth]{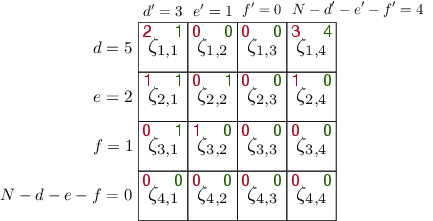}
\caption{An illustration on selecting different combinations for $\zeta_{i,j}$s.}
\label{fig:ZetaTable}
\end{figure} 

\end{example}


\section{Conclusion} \label{Sec:Conc}

In this paper we introduced a code-design framework for multi-rack distributed data storage networks. In this model the encoded data is stored in storage nodes distributed over multiple racks. Each rack has a process unit which is responsible for all calculations and transmissions inside the rack. Practically, the cost of data transmission within a rack is much less than the data transmission across the racks. Therefore, it will be a significant reduction in the repair bandwidth if the node failures are repaired locally inside the racks. Our multi-rack distributed storage code is defined with three parity check matrices $\Ha$, $\Hb$, and $\solG$. We proposed a code-design framework for multi-rack storage networks which is able to locally repair the node failures within the rack using the parity check matrix $\Ha$ in order to minimise the repair cost. However, under the severe failure circumstances where the failures are not repairable by only the survived nodes inside the rack, our coding scheme is capable of engaging the storage nodes in other racks in helping the surviving nodes inside the failed rack to repair the failures where the parity check matrices $\solG$ and $\Hb$ determine the helper racks and the helper nodes, respectively. We justify that the parity matrix $\solG$ can be designed separately from the matrices $\Ha$ and $\Hb$ where indeed it can be designed as the parity check matrix of a locally repairable code such that only a small group of racks need to participate in an inter-rack repair. We established the relation between the rate and the size of the multi-rack storage code and showed that maximising the rate of our multi-rack storage code is equivalent to maximising the code size. In order to maximise the code size, we established a linear programming problem based on the code-design framework criteria. These bounds characterise the trade-off between different code parameters such as inter- and intra-rack resilience and locality. We also exploited symmetries in the linear programming problem for simplifying and significant reductions in its complexity.

\bibliographystyle{IEEEtran}
\bibliography{Ali_bib}

\appendices

\section{Proof of Theorem \ref{thm:2}} \label{append:A2}

By Criterion 1), there exists a row vector ${\bf u}$ of length $S_{2}$ such that 
\begin{align} \label{eq:Kspan}
\r_{j} =  {\bf u} \Hb. 
\end{align}
In other words, $\r_j$ is a vector from row space of $\Hb$. From Lemma \ref{lem:lemma1_4}, as $\beta_{j} \in \Omega(\Ha, \r_{j} , j)$ by  Criterion 3, there exist row vectors ${\bf y}$,  ${\bf y}'$ and $a\in \FFq$ such that 
\begin{align} \label{eq:eVector}
{\bf e}_{j} = {\bf y} \Ha + {\bf y}' \erasurematrix^{\beta_{j}} + a  \r_{j}.
\end{align}
Now, notice
\begin{align}
X_{1,j} & = {\bf e}_{j} X_{1,*}^{\top} \\
& = ({\bf y}  \Ha + {\bf y}'  \erasurematrix^{\beta_{j}} + a  \r_{j}) X_{1,*}^{\top} \\
& =   {\bf y}'  \erasurematrix^{\beta_{j}}  X_{1,*}^{\top}   + a \r_{j} X_{1,*}^{\top} 
\end{align}
where the last equality follows from \eqref{eq:IRP}. Let 
\begin{align} \label{eq:cVector}
[c_{j,1} , \ldots, c_{j,N} ] = {\bf y}'  \erasurematrix^{\beta_{j}}.
\end{align}
Then  
\[
{\bf y}'  \erasurematrix^{\beta_{j}}  X_{1,*}^{\top}  = 
\sum_{n\in\beta_{j}} c_{j,n} X_{1,n}.
\]
Consequently, 
\begin{align}
X_{1,j} & = \sum_{n\in\beta_{j}} c_{j,n} X_{1,n} + a \r_{j} X_{1,*}^{\top}.
\end{align}

Let ${\bf f}_\ell  =[f_{1} , \ldots, f_{M}], ~ \ell \in \{1,\ldots,M\}$, be a length $M$ row vector such that 
\begin{align} \label{eq:fVector}
f_{i} = 
\begin{cases}
1 & \text{ if } i=\ell  \\
0 & \text{ otherwise.}
\end{cases}
\end{align}
Then
\[
X_{1,*} = {\bf f}_1 \X
\]
and hence
\begin{align}
a \r_{j} X_{1,*}^{\top} & = a  X_{1,*} \r_{j}^{\top}  \\
& = a  {\bf f}_1 \X \r_{j}^{\top}.  
\end{align}
From Lemma \ref{lem:lemma1_4}, since $\tau \in \Omega(\solG, 1)$, there exist vectors ${\bf z}$ and ${\bf z}'$ such that 
\begin{align} \label{eq:fVector2}
{\bf f}_1 = {\bf z}  \solG  + {\bf z}' \erasurematrix^{\tau} .
\end{align}
Now, notice that $\erasurematrix^{\tau}$ is a $M\times M$ matrix over $\FFq$, as $\solG$ has only $M$ columns. 
Consequently, 
\begin{align}
a \r_{j} X_{1,*}^{\top} & = a ( {\bf z}  \solG  + {\bf z}' \erasurematrix^{\tau}  )\X \r_{j}^{\top} \\
& =  a {\bf z}' \erasurematrix^{\tau} \X \r_{j}^{\top}  
\end{align}
where the last equality follows from that
\begin{align}
 {\bf z}  \solG\X \r_{j}^{\top} & = {\bf z}  \solG \X  \Hb^{\top}  {\bf u}^{\top}   \\
& =  {\bf z} \left(  \solG \X  \Hb^{\top}  \right) {\bf u}^{\top}  \\
& = 0,
\end{align}
where the last equality follows from \eqref{eq:ARP}.

Let the matrix 
\begin{align} \label{eq:dmatrix}
\left[
\begin{array}{ccc}
d_{j,1,1} & \cdots & d_{j,1,N} \\
\vdots & \ddots &\vdots \\
d_{j,M,1} & \cdots, & d_{j,M,N}  \\
\end{array}
\right] & = a({\bf z}'   \erasurematrix^{\tau}  )^{\top} \r_{j} \\
& = a \erasurematrix^{\tau}  ({\bf z}')^{\top}    \r_{j},
\end{align}
as $\erasurematrix^{\tau} $ is a diagonal matrix. If 
$\r_{j} = [r_{j,1} , \ldots, r_{j,N}] $ and 
${\bf z}'   \erasurematrix^{\tau}  \triangleq {\bf v} \triangleq  [v_{1}, \ldots, v_{M}  ]$, then it can be verified easily that 
\begin{enumerate}
\item $d_{j,m,n} = a v_{m} r_{n}$ and hence 
 $d_{j,m,n} = 0$ if either $m \not\in \tau$ or $n \not \in \mu_{j}$.

\item 
\begin{align}
a \r_{j} X_{1,*}^{\top} & = \sum_{m\in\tau} \left( 
                                        \sum_{n\in \mu_{j}} d_{j,m,n} X_{m,n}
  			 	    \right).
\end{align}
\end{enumerate}
Thus the theorem is then proved. 

\section{Proof of Theorem \ref{thm:cost}} \label{append:A3}

First, we will consider intra-rack transmissions. In each helper rack (say $m \in  \tau$) involved in the inter-rack repair,  $|\mu_{j}|$ nodes in it will need to transmit its data ($X_{m,s}$ where $s\in\mu_{j}$) to the processing unit in the helper rack.  So, the set of symbols sent from the nodes to the processing unit is
$
\bigcup_{j\in\pattern\setminus \alpha}\mu_{j}
$.
Consequently, the total number of symbols sent to the processing units of  helper rack is equal to 
$
|\tau||\bigcup_{j\in\pattern\setminus \alpha}\mu_{j}|
$ where $|\tau|$ is the number of helper racks. That explains the first term in LHS of \eqref{thm3:eqa}.
Next, at each helper rack $m$, it will need to transmit 
\[
\{ \r_{j} X_{m,*}  :  j \in  \pattern\setminus \alpha\}
\]
to the failing rack. 
As these symbols may be linearly dependent, the actual number of symbols that really needed to be transmitted is only $\dim \langle \r_j , j\in \pattern\setminus\alpha \rangle$. Therefore, the total number of inter-rack transmission is $|\tau|\dim \langle \r_j , j\in \pattern\setminus\alpha \rangle$. This explains the upper bound on \eqref{thm3:eqb}.

After receiving the transmissions from the helper racks, the processing unit of the failed rack can now aim to recover the failed nodes. For each $j\in   \pattern$, it requires transmission from nodes in the set $\beta_{j}$. Therefore, the number of intra-rack transmissions in rack 1 from nodes to the processing unit is 
$
|\cup_{j\in\pattern} \beta_{j}|
$. This corresponds to the second term in RHS of \eqref{thm3:eqa}. 
Finally, receiving all the symbols, the processing nodes can recover the contents of all the failed nodes. It requires $|\pattern|$ intra-rack transmissions from the processing unit of the failed rack to the failed nodes for recovery, explaining the last term of LHS of  \eqref{thm3:eqa}. The theorem thus proved.

\section{Proof of Theorem \ref{thm:CodeRate}} \label{append:A4}

According to \eqref{eq:IRP} and \eqref{eq:ARP}, the total number of parity check equations is at most $MS_{1} + LS_{2}$ while  the number of variables is $MN$.
Therefore, the rate of the code is at least
\[
\frac{MN - MS_{1} - LS_{2}}{MN}.
\]

Now, Consider the following set
$$
\S_m ({\b}_{m})=
\left\{
X_{m,*}^{\top} : \: \Ha X_{m,*}^{\top} = \zero \text{ and } \Hb X_{m,*}^{\top} = {\b}_{m}
\right\}.
$$
In other words, $\S_m({\b}_{m})$ is the set of solutions for the system of linear equations specified by parity check matrices $\Ha$ and $\Hb$ in \eqref{eq:IRP} and \eqref{eq:ARP}, respectively. From linear algebra, the number of these solutions for any $\b_m$ is given by
\[
|\S_m({\b}_{m})| = q^{N - S_{1} - S_{2}},
\]
where $q$ is the field size. Similarly, let  
\begin{align}
S({\bf B}) = \left\{\X :\: \Ha X^{\top} = \zero \text{ and } \Hb X^{\top} = {\bf B}\right\}.
\end{align}
where ${\bf B} =  [{\b}_{1}, \ldots, {\b}_{M}]$.
Then 
$
|S({\bf B})| = q^{M(N - S_{1} - S_{2})}
$
for any ${\bf B}$. 

Now, consider the set
\[
{\cal S} \triangleq \{  
\X:\: \X \text{ satisfies \eqref{eq:IRP} and \eqref{eq:ARP}}
\}.
\]
Then it is clear that 
\begin{align*}
|{\cal S}| & = \sum_{{\bf B} : \: {\bf B}\solG^{\top} = \zero  } |S({\bf B})|  \\
& = |\Delta| q^{M(N - S_{1} - S_{2})}
\end{align*}
where
$
\Delta = 
\left\{
\bf B: \: \bf B \solG^{\top} = \zero
\right\}.
$

As the rank of $\solG$ is $L$ and $\bf B$ is a matrix of size $S_{2}\times M$, 
$
|\Delta| = q^{(M-L)S_{2}} 
$. 
Consequently, 
\begin{align}
|{\cal S}|  & = q^{(M-L)S_{2}}q^{M(N - S_{1} - S_{2})} \\
& = q^{MN - MS_{1} - LS_{2}}.
\end{align}
The theorem then follows.

\section{Proof of Proposition \ref{prop2}} \label{append:A5}

By definitions,  
\begin{align}
a^*_{\w,\s} & = \frac{1}{|S_\N|} \sum_{\sigma \in S_\N} a_{\w,\s}^{(\sigma)} \\
a^*_{\w',\s'} & = \frac{1}{|S_\N|} \sum_{\sigma \in S_\N} a_{\w',\s' }^{(\sigma)}.
\end{align}
Now, by   \eqref{eq35}--\eqref{eq37}, there exists a permutation $\mu \in S_{\N}$ such that 
$\w'   =\mu(\w) $ and $\s'  =\mu(\s)$.
Hence, for any permutation $\sigma \in S_{\N}$, we have
\begin{align}
\sigma(\w') & = \sigma(\mu (\w) )  = (\sigma \circ \mu) (\w)  
\end{align}
Similarly, we have 
$
\sigma(\s')  = (\sigma \circ \mu) (\s)  
$.
Consequently, 
\begin{align}
a^*_{\w',\s'} & = \frac{1}{|S_\N|} \sum_{\sigma \in S_\N} a_{\w',\s'}^{(\sigma)} \\
& = \frac{1}{|S_\N|} \sum_{\sigma \in S_\N} a_{\sigma(\w'), \sigma(\s')} \\
& = \frac{1}{|S_\N|} \sum_{\sigma \in S_\N} a_{(\sigma\circ \mu) (\w), (\sigma\circ \mu)(\s)} \\
& = \frac{1}{|S_\N|} \sum_{\sigma \in S_\N} a_{ \sigma (\w),  \sigma (\s)} \\
& = a^*_{\w,\s}
\end{align}
where the second last equality follows from the fact that $S_{\N}$ is a group and hence 
$$
\{ \sigma \circ \mu: \sigma \in S_{\N}\}
=
\{ \sigma : \sigma \in S_{\N}\}.
$$  
Similarly, we can also prove that $c^*_{\w,\s}  = c^*_{\w', \s'} $.

\section{Proof of Theorem \ref{thm5}} \label{append:A}


Recalling our assumptions in Section \ref{Sec:V}, any codeword from the rack model storage code $\C$ is in the form of $(\bx,\by)\in \C$ 
satisfying \eqref{eq:TwoRackCode}. The support enumerator of the dual code $\Cd$ in \eqref{eqThm:a}  follows from the MacWilliam's identity \cite{macwilliams}. 
Equation \eqref{eqThm:b} follows directly  from definition in \eqref{eq:TwoRackCode}. 
From the intra-rack resilience property in Definition \ref{def2},  any $\delta_1$ simultaneous failures in a rack can be repaired via intra-rack repairs.
%
%
Now consider $\w \subseteq \N$ such that $1 \le |\w| \le \delta_1$. Suppose to the contrary that \eqref{eqThm:c} does not hold, i.e., $\Lambda_{\C}(\w,\s) \ge 1$ for some $\s \subseteq \N$. Assume without loss of generality that $\w = \{1, \ldots, |\w|\}$.
Then there exists a codeword $\c = [x_{1}, \ldots, x_{N}, y_1 , \ldots y_N]$ such that 
\begin{enumerate}
\item
$[x_{1}, \ldots, x_{N}] \neq [0, \ldots, 0]$ and hence $[x_{1}, \ldots, x_{N}, y_1 , \ldots y_N] \neq [0, \ldots, 0]$
\item 
$x_{i} =  0 $ for all $N \ge i \ge | \w| +1$
\end{enumerate}
Hence, if the nodes $(x_{1},\ldots,x_{\delta_1})$ fail, then it is impossible to perfectly recover $x_i $ for all $N\ge i \ge |\w| +1$ since the failing rack  cannot distinguish whether the stored symbol was $[x_{1}, \ldots, x_{N}]$ or $[0, \ldots, 0]$. Therefore, we prove \eqref{eqThm:c}.

Follows from the inter-rack resilience property in Definition \ref{def2}, any $\delta_2$ simultaneous failures in one rack can be recovered by the survived nodes in the same rack and all the nodes in the other rack.
Again, suppose to the contrary that $\Lambda_\C(\w,\emptyset) \geq 1$ for some $\w \subseteq \N$ such that $1 \leq |\w| \leq \delta_2$. Assume without loss of generality that $\w= \{1, \ldots, |\w|\}$. Then there exists a codeword $\c = [x_{1}, \ldots, x_{N},0, \ldots 0]$ such that 
\begin{enumerate}
\item
$[x_{1}, \ldots, x_{N}] \neq [0, \ldots, 0]$ and hence $[x_{1}, \ldots, x_{N},0,\ldots,0] \neq [0, \ldots, 0]$
\item 
$x_{i} =  0 $ for all $N \ge i \ge |\w| +1$
\end{enumerate}
Hence, if nodes $x_{1},\ldots,x_{\delta_2}$ fail, then it is impossible to perfectly recover $x_{i} $ for all $N\ge i \ge |\w| +1$. Thus, we prove \eqref{eqThm:e}. 



From the intra-rack locality property,  if the number of failures in a rack is at most  $\Gamma_{1} + 1$, then  each node can be repaired via intra-rack repair, involving at  most $r_{1}$ surviving nodes. 
Now, suppose that the set of failing nodes is $\pattern \cup \{i\}$ where $|\pattern| = \Gamma_{1}$. In order to repair the failed nodes $x_{i}$ via intra-rack repair, there must exists a dual codeword  $(\f' , {\bf 0}) \in \Cd$ whose support is $(\w',\emptyset)$ for some  $\w' \in \Theta_{1}(i,\gamma,r_{1})$. Hence, 
\begin{align} 
\sum_{\w \in \Theta_{1}(i,\gamma,r_1)} 
\Lambda_{\Cd}(\w,\emptyset) \geq 1 
\end{align}
In fact, since supports of $(\f' , {\bf 0})$ and $(c \cdot \f' , {\bf 0})$ are the same for all non-zero $c \in \FF_q$, we have 
\begin{align} 
\sum_{\w \in \Theta_{1}(i,\gamma,r_1)} 
\Lambda_{\Cd}(\w,\emptyset) \geq q-1
\end{align}
and \eqref{eqThm:d} is proved.


Finally, from the inter-rack locality property,  if the number of failures in a rack is at most  $\Gamma_{2} + 1$, then  each node can be repaired via inter-rack repair such that involving $i$) at most  $r_{2}$ surviving nodes in the failing rack and  $ii$) at most  $a$ nodes from each   helper rack. Now, suppose that the set of failing nodes is $\pattern \cup \{i\}$ where $|\pattern| = \Gamma_{2}$. In order to repair the failed nodes $x_{i}$ via intra-rack repair, there must exists a dual codeword  $(\f' , \g') \in \Cd$ whose support is $(\w',\s')  \in \Theta_{2}(i,\gamma,r_{2},a)$. Again, as support of  $(\f' , \g')  $ and 
$(c \f' ,  c \g')  $ for all non-zero $c \in \FF_q$, \eqref{eqThm:f} thus follows.


\section{Proof of Theorem \ref{thm4}} \label{append:B}

The simplified bound is obtained by respectively replacing $A_{\w,\s}$ and $C_{\w,\s}$
with newly introduced variables $X_{d,e,f}$ and $Y_{d,e,f}$ 
where 
\begin{align}
d  = |\w \setminus \s|,  \quad e = |\w \cap \s|, \text{ and } f  = |\s \setminus \w|.
\label{eq38} 
\end{align}
The replacement is possible, due to Proposition \ref{prop2}. 

With the new notations, some terms will become alike and can be grouped together.
For example, the term $\sum_{\w,\s \subseteq \N} A_{\w,\s}$ can be replaced by $\sum_{d,e,f}  \binom{N}{d,e,f} X_{d,e,f}$. Moreover,  some constraints will become equivalent and hence can be omitted. Most constraints in (LP1) can be rewritten directly. The more complicated ones are (C3), (C8) and (C9). In the following, we will illustrate how to simplify and rewrite them.
 
We now first ``reduce'' constraint (C8).  
The key idea of reduction is described as follows: 
First, by Proposition \ref{prop2},
$
C_{\w,\emptyset} = C_{\w',\emptyset}
$
if  $|\w| = |\w'|$.  In fact, by \eqref{eq38}, 
\[
C_{\w,\emptyset} = Y_{|\w|,0,0 }
\]
Consider any given  $(i,\gamma) \in \Phi(\Gamma_{1})$. 
Recall that 
\begin{align}
\Theta_{1}(i,\gamma,r_{1})  
\triangleq 
\{ 
\w  : \: i\in\w, \w \cap  \gamma  = \emptyset  \text{ and } |\w| \le r_{1}+1\}.
\end{align}
Let $\Delta_{2}(d) = |\{\w: \: \w \in \Theta_{1}(i,\gamma,r_{1}) \text{ and } |\w| = d\}|$. 
By direct counting, we can prove that 
\begin{align}\label{eqDelta2}
\Delta_{2}(d) =\binom{N-\Gamma_{1}-1}{d-1} 
\end{align}
Note that $\Delta_{2}(d)$ is independent of the choice of $i$ and $\gamma$. 
Then, the collection of constraint (C8) is reduced to one single constraint
\begin{align}
\displaystyle
\sum_{d=2}^{r_1+1} \Delta_{2}(d) Y_{d,0,0} \ge (q-1) \sum_{d,e,f}  \binom{N}{d,e,f} X_{d,e,f}.
\end{align}

The reduction of constraint (C9) is very similar to that of (C8).
Again, consider any given $(i,\gamma) \in \Omega(\Gamma_{2})$. Recall that 
\[
\Theta_{2}(i,\gamma, r_{2} , a)
 \triangleq 
\{ 
(\w,\s)  : \: i\in \w, \w \cap  \gamma  = \emptyset , 
|\w| \le r_{2}+1 \text{ and } |\s | \le a
\}
\]
Let 
\begin{align*}
\Pi(d,e,f) = 
\left\{
(\w,\s): \: 
\begin{array}{c}
(\w,\s) \in \Theta_{2}(i,\gamma,r_{2},a) \\ \text{ and }
\Upsilon(\w,\s) = (d,e,f) 
\end{array}
\right\},
\end{align*} 
where $\Upsilon(\w,\s)=(|\w \setminus \s|, |\w \cap \s|, |\s \setminus \w| )$. Hence, the constraint (C9) can be rewritten as 
\begin{align}
\displaystyle
\sum_{d+e \le r_{2}+1 , e+f \le a} \Delta_{3}(d,e,f) Y_{d,e,f} \ge (q-1) \sum_{d,e,f}  \binom{N}{d,e,f} X_{d,e,f},
\end{align}
where $\Delta_{3}(d,e,f) = |\Pi(d,e,f)|$.

Finally, it remains to determine $\Delta_{3}(d,e,f)$. 
Partition $\Pi(d,e,f)$ into two sets 
\[
\Pi^{(1)}(d,e,f) \triangleq \{(\w, \s) \in \Pi(d,e,f) \text{ and } i \in \s\}
\]
and 
\[
\Pi^{(2)}(d,e,f) \triangleq \{(\w, \s) \in \Pi(d,e,f) \text{ and } i \not\in \s\}
\]
Hence, 
\begin{align}\label{eqDelta3}
\Delta_{3}(d,e,f)  = |\Pi^{(1)}(d,e,f)| + |\Pi^{(2)}(d,e,f)|
\end{align} 
where, by direct counting,  
\begin{align}
\displaystyle
|\Pi^{(1)}(d,e,f)| =  
\binom{N-\Gamma_{2}-1}{e-1}
  \binom{N-\Gamma_{2}-e-1}{d}
 \binom{N-e-d}{f},  \\
|\Pi^{(2)}(d,e,f)| =  
\binom{N-\Gamma_{2}-1}{e}
 \binom{N-\Gamma_{2}-e-1}{d-1}
 \binom{N-e-d}{f}. 
\end{align}

Finally, we will consider the reduction of   (C3). The approach is similar,   despite that the derivation is much more tedious. First, we fix a given $(\w,\s)$. Consider the term 
\[ 
A_{\w',\s'}   \prod_{j\in \N}  \kappa_q (w'_j ,w_j)\kappa_q (s'_j,s_j)
\]
in (C3).
Let 
\begin{align}
\Xi_{1}(\w,\s) & = \{ i\in\N :\: i\in \w \setminus \s\} \\
\Xi_{2}(\w,\s) & = \{ i\in\N :\: i\in \w \cap \s\} \\
\Xi_{3}(\w,\s) & = \{ i\in\N :\: i\in \s \setminus \w\} \\
\Xi_{4}(\w,\s) & = \{ i\in\N :\: i\in \N \setminus (\w \cup \s)\}.
\end{align}
and
$
\xi_{i,j}(\w,\s,\w',\s') = |\Xi_{i}(\w,\s) \cap \Xi_{j}(\w',\s')|.
$

Suppose 
$
\xi_{i,j}(\w,\s,\w',\s') = \zeta_{i,j}$ for all $ 1\le i,j \le 4.
$ Let $\zeta = (\zeta_{i,j} , \: 1\le i,j \le 4 )$.
Then 
\begin{enumerate}
\item
$
\sigma_{1}(\zeta) = |\w \setminus \w'| + |\s \setminus \s'| 
$
where 
\begin{align} \label{eq:sigma1zeta}
\sigma_{1}(\zeta) 
& \triangleq  \zeta_{1,3} + \zeta_{1,4}  + \zeta_{2,3}  + \zeta_{2,4} 
+ 
\zeta_{2,1} + \zeta_{2,4} + \zeta_{3,1} + \zeta_{3,4}.
\end{align}
Similarly, 
$
\sigma_{2}(\zeta) = |\w \cap \w'| + |\s \cap \s'| 
$
where 
\begin{align} \label{eq:sigma2zeta} 
\sigma_{2}(\zeta)  
& \triangleq  \zeta_{1,1} + \zeta_{1,2}  + \zeta_{2,1}  + \zeta_{2,2} 
+ 
\zeta_{2,2} + \zeta_{2,3} + \zeta_{3,2} + \zeta_{3,3}.
\end{align}

\item
$\prod_{j\in \N} \kappa_q (w'_j ,w_j)\kappa_q (s'_j,s_j)  = (q-1)^{\sigma_{1}(\zeta)} (-1)^{\sigma_{2}(\zeta)}$. 

\item 
$
A_{\w',\s'} = X_{d',e',f'}
$
where 
\begin{align}
d' & = \zeta_{1,1} + \zeta_{2,1} +\zeta_{3,1} +\zeta_{4,1} \\
e' & = \zeta_{1,2} + \zeta_{2,2} +\zeta_{3,2} +\zeta_{4,2} \\
f' & = \zeta_{1,3} + \zeta_{2,3} +\zeta_{3,3} +\zeta_{4,3}  
\end{align}

\item 
Fix $\w,\s$. 
Let  
$
U(\zeta) = 
\left| \{
(\w',\s') :\: \xi_{i,j}(\w,\s,\w',\s') =  \zeta_{i,j},1 \le i,j \le 4
\}\right|.
$
Then 
\begin{align} \label{eq:Uzeta}
U(\zeta ) = 
 \binom{\sum_{j=1}^{4}\zeta_{1j}}{\zeta_{1j}, j =1,2,3,4}
 \binom{\sum_{j=1}^{4}\zeta_{2j}}{\zeta_{2j}, j =1,2,3,4}
 \binom{\sum_{j=1}^{4}\zeta_{3j}}{\zeta_{3j}, j =1,2,3,4}
 \binom{\sum_{j=1}^{4}\zeta_{4j}}{\zeta_{4j}, j =1,2,3,4}.
\end{align}

\item
For any 
$(d,e,f, d',e',f')$, let 
\begin{align}\label{eqDelta1}
\Delta_{1}(d,e,f, d',e',f') = 
\sum_{\zeta: \zeta^{(1)}  = (d,e,f) , \zeta^{(2)}  = (d',e',f') }  U(\zeta)(q-1)^{\sigma_{1}(\zeta)} (-1)^{\sigma_{2}(\zeta)}
\end{align}
where 
\begin{align*}
\zeta^{(1)}  & = \left(\sum_{j} \zeta_{1,j},  \sum_{j} \zeta_{2,j}, \sum_{j} \zeta_{3,j} \right) \\
\zeta^{(2)}  & = \left(\sum_{j} \zeta_{j,1},  \sum_{j} \zeta_{j,2}, \sum_{j} \zeta_{j,3} \right) .
\end{align*}
\end{enumerate}
Grouping all the like terms, the constraint (C3) can be rewritten as in (D3).

\end{document}